\documentclass{article}
\usepackage{fullpage}
\usepackage{amsfonts,amssymb,amsmath,mathtools}
\usepackage{graphicx}
\usepackage[utf8]{inputenc}
\usepackage[english]{babel}

\usepackage{amsthm}
\usepackage[inline]{enumitem}
\usepackage[vlined,ruled]{algorithm2e}

\usepackage{hyperref}
\usepackage{cleveref}
\usepackage{xcolor}
\usepackage{appendix}
\usepackage{amsthm}
\usepackage{mathdots, latexsym, amsbsy, bm, mathtools}
\usepackage{thmtools}
\usepackage{thm-restate}

\usepackage{tikz}
\usetikzlibrary{calc, arrows, automata, graphs, shapes, petri, decorations.pathmorphing, decorations.markings}
\tikzset{->-/.style={decoration={markings,mark=at position .5 with {\arrow{>}}},postaction={decorate}}}
\tikzset{>=latex'}
\tikzset{box/.style args={#1and#2}{draw=#1,rounded corners,fill=#2}}
\tikzset{vertex/.style={draw,circle,inner sep=0pt,minimum size=15pt}}
\tikzset{bvertex/.style={draw,rectangle,rounded corners,inner sep=5pt,minimum height=20pt}}
\tikzset{small vertex/.style={draw,circle,inner sep=0pt,minimum size=10pt}}
\tikzset{l/.style={font=\small}}

\definecolor{rwthblue}{cmyk}{1, .5, 0, 0}
\definecolor{rwthlightblue}{cmyk}{.45, .14, 0, 0}
\definecolor{rwthred}{cmyk}{.15, 1, 1, 0}
\definecolor{rwthgreen}{cmyk}{.7, 0, 1, 0}
\definecolor{rwthorange}{cmyk}{0, .4, 1, 0}
\definecolor{rwthmagenta}{cmyk}{0, 1, .25, 0}
\definecolor{rwthyellow}{cmyk}{0, 0, 1, 0}
\definecolor{rwthpetrol}{cmyk}{1, .3, .5, .3}
\definecolor{rwthturkis}{cmyk}{1, 0, .4, 0}
\definecolor{rwthmay}{cmyk}{.35, 0, 1, 0}
\definecolor{rwthbordeaux}{cmyk}{.25, 1, .7, .2}
\definecolor{rwthviolet}{cmyk}{.7, 1, .35, .15}
\definecolor{rwthlila}{cmyk}{.6, .6, 0, 0}

\newtheorem{theorem}{Theorem}
\newtheorem{lemma}[theorem]{Lemma}
\newtheorem{claim}[theorem]{Claim}
\newtheorem{corollary}[theorem]{Corollary}
\newtheorem{proposition}[theorem]{Proposition}
\newtheorem{observation}[theorem]{Observation}
\theoremstyle{definition}
\newtheorem{example}[theorem]{Example}

\let\epsilon\varepsilon
\let\emptyset\varnothing
\let\tilde\widetilde

\newcommand{\abs}[1]{\left\lvert #1 \right\rvert}
\newcommand{\set}[1]{\left\{#1 \right\}}

\newcommand{\jb}{{j^\bullet}}
\newcommand{\dnew}{d^{\text{new}}}
\newcommand{\dnewb}{\f{d}^{\text{\textbf{new}}}}
\newcommand{\bnew}{b^{\text{new}}}
\newcommand{\bnewb}{\f{b}^{\text{\textbf{new}}}}
\newcommand{\pnew}{p^{\text{new}}}
\newcommand{\pnewb}{\f{p}^{\text{\textbf{new}}}}
 
\newcommand{\f}{\boldsymbol}
\newcommand*\samethanks[1][\value{footnote}]{\footnotemark[#1]}

\newenvironment{subproof}[1][\proofname]{\vspace{-3ex}\begin{proof}[#1]}{\end{proof}}

\DeclareMathOperator\val{val}
\DeclareMathOperator\capa{cap}

\newcommand{\Mnat}{$\mathrm{M}^\natural$}
\newcommand{\Lnat}{$\mathrm{L}^\natural$}
\newcommand{\DO}{\mathsf{DO}}
\newcommand{\ExO}{\mathsf{ExO}}
\newcommand{\TO}{\mathsf{TO}}

\title{A flow-based ascending auction to compute buyer-optimal Walrasian prices\footnote{This work has been published in Networks under open access \url{http://doi.org/10.1002/net.22218}.}}

\date{\vspace{2ex} April 18, 2024}
\author{
	Katharina Eickhoff\thanks{School of Business and Economics, RWTH Aachen University. Email: \{eickhoff, peis, rieken\}@oms.rwth-aachen.de.}
    \footnote{K. Eickhoff is funded by the Deutsche Forschungsgemeinschaft 
		(DFG, German Research Foundation) – 2236/2}\and
	S. Thomas McCormick\footnote{Sauder School of Business, University of British Columbia. Email:  thomas.mccormick@sauder.ubc.ca} \and
	Britta Peis\samethanks[2] \and
	Niklas Rieken\samethanks[2] \and
	Laura Vargas Koch\thanks{Research Institute of Discrete Mathematics, Uni Bonn, \texttt{lvargas-koch@dm.uni-bonn.de}}
}

\begin{document}
	
	\maketitle
	
	\begin{abstract}
		We consider a market where a set of objects is sold to a set of buyers, each equipped with a valuation function for the objects. 
		The goal of the auctioneer is to determine reasonable prices together with a stable allocation.
		One definition of ``reasonable'' and ``stable'' is a Walrasian equilibrium, which is a tuple consisting of a price vector together with an allocation satisfying the following desirable properties: \begin{enumerate*}[label=(\roman*)]\item the allocation is market-clearing in the sense that as much as possible is sold, and \item the allocation is stable in the sense that every buyer ends up with an optimal set with respect to the given prices.\end{enumerate*} 
		Moreover, ``buyer-optimal'' means that the prices are smallest possible among all Walrasian prices.
		
		In this paper, we present a combinatorial network flow algorithm to compute buyer-optimal Walrasian prices in a multi-unit matching market with truncated additive valuation functions.
		The algorithm can be seen as a generalization of the classical housing market auction and mimics the very natural procedure of an ascending auction. 
		We use our structural insights to prove monotonicity of the buyer-optimal Walrasian prices with respect to changes in supply or demand.
	\end{abstract}

    \vspace{2ex} 
	\section{Introduction}\label{sec:intro}
	We consider a market where $m$ different indivisible types of \emph{objects} $\Omega = \{i_1, \ldots, i_m\}$ are sold to $n$ \emph{buyers} $N = \{j_1, \ldots, j_n\}$.
	Every object $i \in \Omega$ has a \emph{supply} of $b_i \in \mathbb{Z}_+$ identical copies of that object.
	All buyers $j \in N$ have a \emph{demand} $d_j \in \mathbb{Z}_+$, which is the maximum number of items they want to purchase.
	The goal of the auctioneer is to find a per-unit \emph{price} $p(i)$ for each object $i \in \Omega$ together with an \emph{allocation} $\f{x} \in \mathbb{Z}_+^{\Omega \times N}$ of items to buyers such that the prices $\f{p}$ and the allocation $\f{x}$ satisfy certain desirable properties.
	Certainly, the allocation $\f{x}$ should be \emph{feasible} in the sense that at most $b_i$ units of each object $i \in \Omega$ are sold, and each buyer $j \in N$ gets awarded at most $d_j \in \mathbb{Z}_+$ items.
	We consider \emph{truncated additive valuations} $v_j \colon \mathbb{Z}_+^{\Omega} \to \mathbb{R}_+$, based on a value $v_{ij} \in \mathbb{Z}_+$ that a buyer $j \in N$ has per unit of object $i \in \Omega$, and demand $d_j \in \mathbb{Z}_+$, 
	\[
	v_j(\f{y}) \coloneqq \max_{\f{y'} \leq \f{y}}\set{\sum_{i \in \Omega} v_{ij} y'_{i} \,\middle\vert\, \sum_{i \in \Omega} y'_i \leq d_j} \qquad \text{for all } j \in N, \f{y} \in \mathbb{Z}_+^{\Omega} .
	\]
	The \emph{utility} of one unit of object $i \in \Omega$ to buyer $j \in N$ is then given by $v_{ij} - p(i)$.
    
	A standard and natural equilibrium concept for markets are \emph{Walrasian equilibria}, named after L\'eon Walras \cite{walras1874}.
    To properly define those equilibria we first need some more notation. 
	We call a feasible allocation \emph{stable}, when each buyer $j \in N$ gets a best possible allocation $\f{x}_{\bullet j}$ given prices $\f{p}$, in other words, every buyer is assigned to one of her \emph{preferred bundles}, i.e.\ a vector from the set
	\begin{equation}\label{eq:demand-set}
		D_j(\f{p}) = \arg\max_{\f{x}_{\bullet j}} \left\{\sum_{i \in \Omega} (v_{ij} - p(i))x_{ij} \,\middle\vert\, \sum_{i \in \Omega} x_{ij} \leq d_j, 0 \leq x_{ij} \leq b_i, i \in \Omega\right\}.
	\end{equation}
	If a price vector $\f{p}$ admits a stable allocation, then $\f{p}$ is called \emph{competitive}.
	Note that competitive prices always exist as $p(i) = 1 + \max_{j \in N} v_{ij}$ for all $i \in \Omega$ is competitive with stable allocation $\f{x} = \f{0}$.
	However, these prices are not interesting.
	We want to find prices that are not only competitive but also \emph{market-clearing}, i.e., as much as possible should be sold, which is $\min\{\sum b_i, \sum d_j\}$ in our model.
	Prices that are competitive and market-clearing are called \emph{Walrasian} and a pair $(\f{p}, \f{x})$ of Walrasian prices and supporting stable allocation is called \emph{Walrasian equilibrium}.
	Walrasian prices where the prices are as low as possible, i.e., where $\sum p(i)$ is minimum among all Walrasian prices, are called \emph{buyer-optimal} Walrasian prices or \emph{minimum Walrasian prices}. 
	
	The definitions above assume that the minimum selling price for an item is $0$. 
	If there are positive minimum selling prices, the notion of market-clearing changes. 
	An allocation is market-clearing if everything with a price above the minimum selling price is sold. 
	We can easily adapt our results for this more general case, but for simplicity we assume minimum selling prices to be $0$ in the following.
	
	Recall that in our model the buyers' valuations $v_j\colon \mathbb{Z}_+^\Omega \to \mathbb{R}$ belong to the class of truncated additive valuations.
	This is an important special case of \emph{gross substitutes} \cite{kelso1982job}, which form the biggest class of valuation functions that are in some sense well-behaved, which will be covered in Subsection~\ref{sec:intro:related-work}.
	A fundamental property is that Walrasian prices are guaranteed to exist if all buyers have gross substitute valuations (see \cite{ausubel2006efficient} for multi-supply and \cite{kelso1982job} for unit-supply settings) and that the (unique) buyer-optimal Walrasian price vector can be found using an ascending auction (\cite{ausubel2006efficient,gul2000english}).
	
	Truncated additive valuations are a natural class of valuation functions.  
	It will turn out that they allow for a simple flow-based algorithm to execute the auction. 
	Moreover, truncated additive valuations generalize the special case of single-unit matching markets.
	
	\paragraph{Single-unit matching markets.}
	The model we consider generalizes the classical matching market model, in the literature also commonly referred to as housing market (see e.g.~\cite{easley2010}).
    In those markets buyers have only unit-demand, i.e., $d_j = 1$ for all $j \in N$, or more formally, $v_j(\f{x}_{\bullet j}) = \max\{v_j(\f{\chi}_i) \mid i \in \Omega \text{ and } x_{ij} > 0\}$, and also each object is only available in one copy, i.e., $b_i = 1$ for all $i \in \Omega$.
    Here, we use $\f{\chi}_X \in \{0,1\}^\Omega$ to denote the characteristic vector of $X$, i.e., the zero-one vector in $\{0,1\}^\Omega$ with a $1$ precisely in the entries corresponding to objects in $X$.
	
	For this setting, a seminal paper by Demange et al. \cite{demange1986} describes an ascending auction which starts at the minimum possible selling price, i.e.\ $p(i) = 0$ for all $i \in \Omega$, and then iteratively raises the prices on a set of items that are overdemanded (which is a set of neighbors of a Hall set, i.e., a subset of buyers that are jointly interested in a set of items without enough supply for the demand in the subset) until the prices are market-clearing.
	If one only raises prices on an inclusion-wise minimal overdemanded set, they show that this process yields component-wise minimum (in fact unique) competitive prices, and that these prices are market-clearing, and hence, Walrasian.
	It is also well known that one can determine a socially optimal allocation together with buyer-optimal market-clearing prices with an adaption of the Hungarian algorithm, a primal-dual algorithm (see~\cite{shapley1971assignment}).
	For a good overview on these topics, we refer to the book \cite{easley2010} or the short but self-contained paper by Kern et al. \cite{kern2016}.
	
	\paragraph{The duplication method.}
	A na\"ive approach to reduce our more general multi-unit auction to a single-unit auction is via the following \emph{duplication method}: 
	we replace the $b_i$ copies of object $i \in \Omega$ by $b_i$ individual objects with unit-supply, and we replace a buyer $j \in N$ demanding $d_j$ items by $d_j$ unit-demand buyers, with the same valuations. 
	Certainly, an ascending auction of the single-unit instance will return Walrasian prices, but these prices are in general not buyer-optimal:
	\begin{example}\label{ex:1}
		Consider one buyer $N = \{1\}$ with a demand of $d_1 = 2$, and two different items $\Omega = \{\alpha, \beta\}$ with a supply of $b_\alpha = b_\beta = 1$ which are valued differently by the sole buyer, say $v_{\bullet 1} = (5, 1)$.
		If we copy the buyer, there is no stable allocation if $p(\alpha) < 4$ as both copies compete for object $\alpha$. If $p(\alpha) = 4$, both copies of the sole buyer are indifferent between the objects and thus, $\f{p} = (4, 0)$ and $\f{x} = ((1, 0), (0, 1))$ is a Walrasian equilibrium.
		However, considering the original situation, since the buyer is alone $\f{p} = (0, 0)$ and $\f{x} = (1, 1)$ is the buyer-optimal Walrasian equilibrium.
        The intuitive explanation is that the duplication method is oblivious to the fact that the two copies represent the same buyer and hence, the computed Hall sets are wrongly interpreted as competition between two buyers.
		Thus, the prices computed by the {duplication method} are not buyer-optimal.
	\end{example}
	
	As we just saw in the example above, an algorithm to compute buyer-optimal prices in single-unit matching markets does not trivially give an algorithm for buyer-optimal prices when buyers have truncated additive valuations. 
	The good news is that these valuation functions still allow for a nice flow-based algorithm to compute overdemanded sets. 
	This generalizes  the method to use a Hall set which works in the single-unit case and is more intuitive than the submodular function minimization methods used for strong gross substitute valuations.

	\subsection{Our contribution}
	In this paper, we provide a flow-based ascending auction for multi-unit markets where all buyers have truncated additive valuations. 
	This auction enables us to compute buyer-optimal Walrasian prices and a stable allocation.
	It generalizes the matching-based ascending auction introduced by \cite{demange1986}.
	
	Although there already exist ascending auctions that compute the minimum Walrasian prices by e.g. Ausubel \cite{ausubel2006efficient} and Murota et al.\ \cite{murota2013computing}, these algorithms strongly rely on the theory of submodular function minimization and discrete convexity to compute the set of objects on which prices should be increased in the ascending auction. 
	In contrast, we can compute this set with a simple and fast algorithm based on a single max flow min cut computation.
	This enables us to prove all our results independent from the literature on the more general case of auctions where the buyers have gross substitute valuation functions.
	Moreover, it allows us to show that minimum Walrasian prices react to changes in supply and demand in a natural way, i.e., they can only increase when supply decreases or demand increases. 
	This is a very natural monotonicity property, but it was not addressed in the literature yet.

	More concretely, in Section~\ref{sec:flow}, we present an ascending auction which iteratively raises prices on the objects in the left-most (i.e., inclusion-wise minimal) min $s$-$t$ cut in an associated auxiliary flow network.
	In Section~\ref{sec:walrasian}, we prove that the auction indeed terminates with minimum Walrasian prices and a stable allocation.
	Section~\ref{sec:monotonicity} shows how one can obtain price monotonicity results purely from structural insights of the flow-based auction.
	Section~\ref{sec:comparison} compares our work to previous work.

	\subsection{Known results and related work}\label{sec:intro:related-work}
	In this subsection, we give a short summary on what is known about ascending auctions (single- and multi-unit), in particular with valuation functions that are less restrictive.
	We also refer the reader to the excellent survey paper by Paes Leme \cite{leme2017gross}.
	The basic setup remains the same as above, i.e., an auctioneer wants do determine prices $\f{p} \in \mathbb{Z}_+^\Omega$ and an allocation $\f{x} \in \mathbb{Z}_+^{\Omega \times N}$ satisfying $\sum_{j \in N} x_{ij} \leq b_i$ for all $i \in \Omega$.
	However, each buyer $j \in N$ might not have a constant $v_{ij}$ for each object $i \in \Omega$ but instead a valuation function $v_j(\f{x}_{\bullet j})$.
	The net utility buyer $j$ gets given a price vector $\f{p}$ and allocation $\f{x}$ is then $v_j(\f{x}_{\bullet j}) - \f{p}^T \f{x}_{\bullet j}$ and hence, a buyer's preferred bundle under prices $\f{p}$ is given by
	\[
	D_j(\f{p}) \coloneqq \arg\max\{v_j(\f{x}_{\bullet j}) - \f{p}^T \f{x}_{\bullet j} \mid \f{0} \leq \f{x}_{\bullet j} \leq b, \f{x}_{\bullet j} \in \mathbb{Z}^\Omega\}.
	\]
	In the sequel, it is also convenient to define a buyer's \emph{indirect utility function}, the maximum utility a buyer can get under a given price vector:
	\[
	V_j(\f{p}) \coloneqq \max\{v_j(\f{x}_{\bullet j}) - \f{p}^T \f{x}_{\bullet j} \mid \f{0} \leq \f{x}_{\bullet j} \leq b, \f{x}_{\bullet j} \in \mathbb{Z}^\Omega\}.
	\]
	A valuation function $v_j\colon \mathbb{Z}_+^\Omega \to \mathbb{Z}$ is \emph{gross substitute} if for all price vectors $\f{p}, \f{q} \in \mathbb{R}^\Omega$ with $\f{p} \leq \f{q}$ it holds that for all $\f{x}_{\bullet j} \in D_j(\f{p})$ there exists $\f{y}_{\bullet j} \in D_j(\f{q})$ such that $x_{ij} \leq y_{ij}$ for all $i \in \Omega$ with $p(i) = q(i)$.
	If additionally $\sum_{i \in \Omega} x_{ij} \geq \sum_{i \in \Omega} y_{ij}$ holds for these allocations $\f{x}, \f{y}$, then $v_j$ is \emph{strong gross substitute}.
	Intuitively, this condition expresses that if one increases the price on one object, the demand on other objects (whose prices did not increase) does not diminish.
	What makes this class of functions (introduced by Kelso and Crawford \cite{kelso1982job}) particularly interesting is that Walrasian prices are guaranteed to exist if all buyers have gross substitute valuation functions (see also \cite{ausubel2006efficient}).
	This is not true for more general valuation functions, which makes gross substitutes essentially the widest class of valuation functions that can be handled.\footnote{However, we should point out that there are other classes that are in some sense orthogonal to gross substitutes that have different interesting applications. In particular, \emph{complementarities} form one such interesting class.}
	Gul and Stacchetti \cite{gul1999walrasian} proved some equivalent characterizations of gross substitutes (such as the \emph{single improvement} or \emph{no complementarities} condition) which are insightful on their own, but which also make gross substitute more convenient to consider from an algorithmic point of view. 
	Fujishige and Yang \cite{fujishige2003note} also showed that there is a very fundamental connection of gross substitutes to discrete convex analysis, namely, a valuation function is gross substitute if and only if it is \Mnat-concave.
	A few more characterizations of gross substitutes can be found in the more recent work of Ben-Zwi \cite{ben2017walrasian}, which also provides more fundamental insights to ascending auctions and \emph{overdemanded sets}.
	However, the aforementioned publications only handle the single-unit case, i.e., each object is only available in one copy.
	Murota et al. \cite{murota2013computing} showed that these results all transfer when going to the multi-unit setting if the valuation functions are \emph{strong} gross substitutes (going to the multi-unit setting in mathematical terms is to go from the Boolean lattice to the integer lattice as the domain of the valuation functions).
	
	The connection to discrete convex analysis is helpful to define an ascending auction that can find Walrasian prices.
	For the single-unit case, Gul and Stacchetti \cite{gul2000english} laid out the framework for an ascending auction which naturally increases prices on subsets of objects that are overdemanded.
	However, they only showed that overdemanded sets must exist if the current price vector is not Walrasian using matroid theory (and giving a Hall-type condition), they did not show how to efficiently find those sets.
	Ausubel \cite{ausubel2006efficient} and Murota et al. \cite{murota2013computing,murota2016time} also considered auctions for the multi-unit setting.
	Their auctions follow essentially the same idea, i.e., increasing prices on overdemanded sets but also allow different start prices and price reduction steps (on underdemanded sets, which do not occur if the starting prices are low enough and the price increment step is implemented correctly).
	The key contribution of their work is the algorithm to find overdemanded sets efficiently. 
	All of those auctions rely on a potential function (the \emph{Lyapunov} function)
	\[
	L(\f{p}) = \sum_{j \in N} V_j(\f{p}) + \f{p}^T \f{b}.
	\]
	The main features of the Lyapunov function are that it is minimized exactly at Walrasian prices, and that it is \Lnat-convex (and in particular, submodular) if all valuation functions are \Mnat-concave (or equivalently, strong gross substitute).
    A buyer-optimal Walrasian price vector can be found if one selects an (inclusion-wise) minimal minimizer of $X \mapsto L(\f{p} + \f{\chi}_X)$.
	Hence, a Walrasian price vector can be found efficiently by using submodular function minimization but also a steepest descent direction of $L$ (which corresponds to a maximum\footnote{Here \emph{maximum} does not refer to the size of the set but to the \emph{overdemandedness}, i.e., the difference between demand and supply of that set.} overdemanded set) can be computed in strongly polynomial time \cite{murota2013computing}.

    In~\cite{eickhoff2023faster}, a follow-up paper of the paper at hand, it is described how to find the overdemanded sets for general strong gross substitute valuation functions using polymatroid sum. This leads to the so far best known running time $\mathcal{O}(\abs{N} \cdot \DO + \abs{N} \abs{\Omega}^3 \cdot \ExO)$ for one iteration where buyers have gross substitute valuations, and where $\DO$ is the running time of an demand oracle and $\ExO$ is the running time of an exchange oracle. For a detailed description of the oracle models see \Cref{sec:structure}.
    
	Murota et al.~\cite{murota2013computing} also showed that the ascending auction needs $\|\f{p} - \f{p}_0\|$ iterations to terminate, where $\f{p}_0$ is the initial price vector and $\f{p} \geq \f{p}_0$ a Walrasian price vector.
    
    \vspace{2ex} \ 
	
	\section{A flow-based ascending auction}\label{sec:flow}
	We first sketch our flow-based ascending auction, called the \ref{alg:price-raising}.
	The auction starts with the all-zero vector $p_0(i) = 0$ for all $i \in \Omega$ (or with any initial price vector $\f{p}_0$ known to be a lower bound on the minimum Walrasian price vector $\f{p}^*$).
	In each iteration, given the current price vector $\f{p} \in \mathbb{R}_+^{\Omega}$, the algorithm computes an integral $s$-$t$-flow $f$ of maximum value in an auxiliary flow network $G(\f{p})$ (described below). 
	If the value $\val(f)$ of flow $f$ equals the sum of capacities on the $s$-leaving arcs in $G(\f{p})$, denoted as $D_{\f{p}}$, the algorithm stops.
	Otherwise ($\val(f) < D_{\f{p}}$), the prices on all objects in the left-most min cut are raised by one unit, and the algorithm iterates with the updated price vector.
    Note that the notation \emph{left-most} is motivated from the natural way of drawing a layered $s$-$t$-network with the source on the left, the sink on the right and all remaining vertices layer-by-layer in between.
    However, we can also define the left-most min cut without using any kind of topology on the vertices by defining it as the (unique) \emph{inclusion-wise minimal} set $S$, with $s \in S$ that minimizes the sum of capacities on the outgoing arcs, $\capa(S)$. 
    We can compute this set via a breadth-first-search (BFS) from $s$ in the residual network which corresponds to a max flow.
	Then Theorem~\ref{thm:algo_computes_minimal_competitive_prices} shows that the final price vector $\f{p}^*$ returned by the algorithm is the minimum (and thus buyer-optimal) competitive price vector. 
	
	For computing a corresponding stable and market-clearing allocation~$\f{x}^*$ such that $(\f{p}^*, \f{x}^*)$ is a buyer-optimal Walrasian equilibrium, we modify $G(\f{p}^*)$ slightly to network $H(\f{p}^*)$, which allows us to find a stable allocation where $\min\{\sum_{i \in \Omega} b_i, \sum_{j \in N} d_j\}$ items are sold.
	Furthermore, we show that each object $i \in \Omega$ with $p^*(i) > 0$ is completely sold.

	\subsection{Structure of preferred bundles and oracle models} \label{sec:structure}
	For a fixed price vector $\f{p}$ and a fixed buyer $j$, a minimal preferred bundle of buyer $j$ can be computed with the following greedy approach.
	Let $\prec_j$ be a total ordering of the items by non-decreasing payoffs, i.e., satisfying
	\begin{equation}\label{lin-extension}
		v_{ij} - p(i) \geq v_{kj} - p(k) \quad \text{whenever} \quad i \prec_j k.
	\end{equation}
	For ease of notation, assume that $\prec_j = (1, \ldots, m)$. 
	Note that the ordering $\prec_j$ might not be unique due to ties in the payoffs.
	However, each ordering $\prec_j$ satisfying \eqref{lin-extension} uniquely defines a minimum preferred bundle of buyer $j$ which can be constructed as follows:
    
    \begin{algorithm}
		\SetAlgoRefName{Algorithm to construct a minimum preferred bundle for buyer $j$}
		\KwIn{Supplies $b_i$, demand $d_j$, valuations $v_{ij}$ for all $i \in \Omega$ and an ordering $\prec_j = (1, \ldots, m)$ satisfying \eqref{lin-extension}}
		\KwOut{A minimal preferred bundle for buyer $j$ and the last item $k_j$ selected for this bundle}
		Initially, no items are selected, the residual demand $d_j^{res}$ is equal to $d_j$, and $i \coloneqq 1$\\
        \While{$d_j^{res} > 0$, $v_{ij} > p(i)$ and $i\leq m$}{
            Set $x_{ij} \coloneqq \min\{b_i, d_j^{res}\}$\\
            Iterate with residual demand $d_j^{res} \coloneqq d_j^{res} - x_{ij}$, and $i \coloneqq i+1$
        }
        \Return{$x_j$ and $k_j \coloneqq i-1$}
        \caption{}
        \label{alg:preferred_bundles}
	\end{algorithm}
	
	We observe that the greedy algorithm selects exactly $b_i$ copies of all items in $\{1, \ldots, k_j - 1\}$ and $d_j - \sum_{\ell=1}^{k_j-1} b_\ell$ copies from item $k_j$, and that the set of \emph{all} minimal preferred bundles obeys the following structure.
	Consider the two item sets $\Omega'_j$ and $\Omega_j''$ consisting of all items of larger or equal payoff than the payoff of item $k_j$, which was selected last by the greedy algorithm, i.e.,
	\begin{align*}
		\Omega_j'(\f{p}) &\coloneqq \{i \in \Omega \mid v_{ij}-p(i)>v_{k_jj}-p(k_j)\}, \text{ and}\\
		\Omega_j''(\f{p}) &\coloneqq \{i \in \Omega \mid v_{ij}-p(i)=v_{k_jj}-p(k_j)\}.
	\end{align*}
	Note that $\Omega_j'(\f{p})$ can be empty. 
	This happens when all elements in the preferred bundle have the same payoff\footnote{Note that this also holds even if there is no item that is dropped (i.e., the supply for these objects does not exceed the demand). This way, the formal definition is straightforward and we do not have to deal with case distinctions in the analysis.}.
	
	If $\sum_{i \in \Omega_j'(\f{p}) \cup \Omega_j''(\f{p})} b_i < d_j$, a minimal preferred bundle contains less than $d_j$ items. 
	In such a case, and if there are objects $i \in \Omega$ of payoff $v_{ij} - p(i) = 0$, we assume what is called ``free disposal'' in the literature. 
	Namely, that buyer $j$ is indifferent between choosing a preferred bundle as described, or filling the demand up with items in 
	\[
        \Omega_j'''(\f{p}) \coloneqq \{i \in \Omega \mid v_{ij} - p(i) = 0\}.
	\]
	Thus, in every preferred bundle, buyer $j$ buys exactly $d_{j'}(\f{p}) \coloneqq \sum_{i \in \Omega_j' (\f{p})} b_i$ items from objects in $\Omega_j'(\f{p})$ and $d_{j''}(\f{p}) \coloneqq \min \{\sum_{i \in \Omega_j''(\f{p}) } b_i,\ d_j - d_{j'}(\f{p})\}$ items from objects in $\Omega_j''(\f{p})$. 
	In addition, there might be up to $d_{j'''}(\f{p}) \coloneqq \min \{\sum_{i \in \Omega_j'''(\f{p})} b_i,\ d_j - d_{j'}(\f{p}) - d_{j''}(\f{p})\}$ items of objects with zero payoff in a preferred bundle.
	
	To shorten notation, we omit the price when $\f{p}$ is clear from the context, e.g., we write $\Omega_j'$ instead of $\Omega_j'(\f{p})$, or $d_{j'}$ instead of $d_{j'}(\f{p})$.

    \subsubsection{Oracles.}
    Since we consider this market in an auction setting, it is natural (or even necessary) to assume that the auctioneer has some sort of communication protocol with the buyers.
    Ascending auctions, such as those by \cite{ausubel2006efficient,gul2000english}, require the buyers to reveal their demand correspondence $D_j(\f{p})$ in every iteration for some price vector $\f{p}$ (which can have size exponential in $|\Omega|$).
    In the model with truncated additive valuations, we can ask a buyer to provide her demand $d_j$ and her demand sets $\Omega_j'(\f{p})$, $\Omega_j''(\f{p})$, and $\Omega_j'''(\f{p})$ (the latter one is not required for computing overdemanded sets, however).
    Asking the buyers for those sets is quite intuitive as $\Omega'$ denotes the objects which she wants strongly, $\Omega''$ denotes the objects of which she wants to get as many as she can to fill up the demand, and $\Omega'''$ denotes the objects she is indifferent to buy (as buying them results in payoff zero).
    Hence, it is natural in our setting that the auctioneer only asks the buyers for those three sets and her total demand to get a much more compact representation of the demand sets.
    We denote such an oracle call, in other words a question from the auctioneer to one buyer for $\Omega', \Omega'', \Omega'''$ and $d$, as \emph{tier oracle} and use $\TO$ for short. In sequel, in our runtime analysis we include $\TO$ for these oracle calls. Note that a buyer, who knows her valuations can answer the oracle call in $O(\abs{\Omega} \log(\abs{\Omega}))$ (see \ref{alg:preferred_bundles}).

    In the literature often a demand oracle $\DO$ is used (i.e., asking for an minimal preferred bundle) or an exchange oracle $\ExO$ (i.e., asking how many items of one object the buyer is willing to exchange against the same amount of items of a different object). Each of these oracles can be used in turn, and it depends on the objective which one is most efficient to use. In our case the tier oracle is the most suitable.

	\subsection{Construction of auxiliary flow networks}\label{section:aux-network}
	In each iteration of our flow-based ascending auction with current price $\f{p}$, buyer $j$ reveals to the auctioneer the following information: 
	the two  sets $\Omega_j'$ and $\Omega_j''$, together with the amounts $d_{j'}$ and $d_{j''}$ they want to buy from  these sets, respectively, and the set $\Omega_j''' = \Omega_j'''(\f{p})$ of items of payoff $0$ with the amount $d_{j'''}$.
	
	Given this information, our algorithm constructs a network $G(\f{p})$ and uses a max-flow computation to compute the set of objects on which the price has to be increased. 
	In the following we describe the construction of this network. 
	An example with two buyers and three objects is given in \Cref{fig:examplenetwork}.

    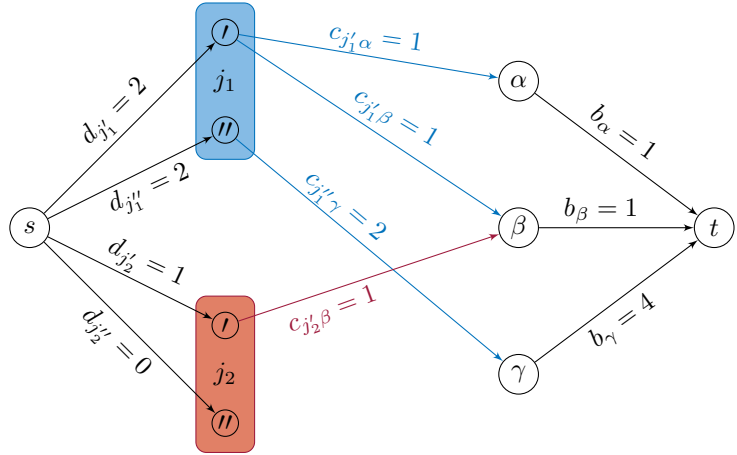
\begin{figure}
		\centering
		\begin{tikzpicture}[>=latex',every node/.style={circle,inner sep=0pt,minimum size=15pt},scale=1.3]
			\node[vertex] (s) at (0, 0) {$s$};
			\draw[box=rwthblue and rwthlightblue] (1.7, .7) rectangle (2.3, 2.3);
			\node (j1) at (2, 1.5) {$j_1$};
			\node[small vertex] (j11) at (2, 2) {$\prime$};
			\node[small vertex] (j12) at (2, 1) {$\prime\prime$};
			\draw[box=rwthbordeaux and rwthred!60] (1.7, -.7) rectangle (2.3, -2.3);
			\node (j2) at (2, -1.5) {$j_2$};
			\node[small vertex] (j21) at (2, -1) {$\prime$};
			\node[small vertex] (j22) at (2, -2) {$\prime\prime$};
			\node[vertex] (i1) at (5, 1.5) {$\alpha$};
			\node[vertex] (i2) at (5, 0) {$\beta$};
			\node[vertex] (i3) at (5, -1.5) {$\gamma$};
			\node[vertex] (t) at (7, 0) {$t$};
			\draw[->] (s) -- (j11) node[rotate=45,pos=.6,above=-3ex]{$d_{j_1'} = 2$};
			\draw[->] (s) -- (j12) node[rotate=25,pos=.5,below=-2.5ex]{$d_{j_1''} = 2$};
			\draw[->] (s) -- (j21) node[rotate=-25,pos=.48,above=-2.5ex]{$d_{j_2'} = 0$};
			\draw[->] (s) -- (j22) node[rotate=-45,pos=.6,below=-3ex]{$d_{j_2''} = 1$};
			\draw[->,rwthblue] (j11) -- (i1) node[rotate=-8,pos=.5,above=-3ex]{$c_{j_1'\alpha }=1$};
			\draw[->,rwthblue] (j11) -- (i2) node[rotate=-31,pos=.5,above=-3ex]{$c_{j_1'\beta }=1$};
			\draw[->,rwthblue] (j12) -- (i3) node[rotate=-40,pos=.28,above=-3ex]{$c_{j_1''\gamma }=2$};
			\draw[->,rwthbordeaux] (j22) -- (i2) node[rotate=30,pos=0.3,below=-2.5ex]{$c_{j_2' \beta}=1$};
			\draw[->] (i1) -- (t) node[rotate=-36,pos=.4,above=-2ex]{$b_{\alpha} = 1$};
			\draw[->] (i2) -- (t) node[pos=.4,above=-2ex]{$b_{\beta} = 1$};
			\draw[->] (i3) -- (t) node[rotate=36,pos=.4,below=-2ex]{$b_{\gamma} = 4$};
		\end{tikzpicture}
		\caption{The network $G(\f{p})$ with two buyers $j_1, j_2$ and three objects $\alpha, \beta, \gamma$. We have valuations $v_{j_1} = (3, 2, 1)$ and $v_{j_2} = (0, 2, 0)$, demands $d_{j_1} = 4, d_{j_2} = 2$, supply $b_\alpha = b_\beta = 1, b_\gamma = 4$, and current prices $p(\alpha) = p(\beta) = p(\gamma) = 0$. In this example, the item sets at prices $\f{p} = (0,0,0)$ are given by $\Omega_1' = \{\alpha, \beta \}$, $\Omega_1'' = \{\gamma\}$, $\Omega_1''' = \emptyset$ and $\Omega_2' = \emptyset$, $\Omega_2'' = \{\beta\}$, $\Omega_1''' = \{\alpha, \gamma\}$.}
		\label{fig:examplenetwork}
	\end{figure}
	
	The vertex set of $G(\f{p})$ consists of a source $s$, a sink $t$, one vertex for each object $i \in \Omega$, and two vertices $j'$ and $j''$ for each buyer $j\in N$.
	We denote the collection of the $j'$ vertices as $N'$ and the one of $j''$ as $N''$.
	Vertices $j'$ and $j''$ correspond to the sets $\Omega_j'(\f{p})$ and $\Omega_j''(\f{p})$, respectively.
 
    \vspace{2ex} \ 
 
	The arcs (with positive capacity) are defined as follows:
	\begin{align*}
		(s, j') & \quad\text{with capacity $d_{j'}$ for all $j \in N$,} \\
		(s, j'') & \quad\text{with capacity $d_{j''}$ for each $j \in N$,} \\
		(j', i) & \quad\text{with capacity $c_{j'i} \coloneqq b_i$ for all $j \in N$, $i \in \Omega_j'$} \\
		(j'', i) & \quad\text{with capacity $c_{j''i} \coloneqq \min\{b_i, d_{j''}\}$ for all $j \in N$, $i \in \Omega_j''$, and} \\
		(i, t) & \quad\text{with capacity $b_i$ for all $i \in \Omega$.}
	\end{align*}
	
	We denote the total capacity on the $s$-leaving arcs by $\capa(s)$, and observe that 
	$\capa(s) = \sum_{j \in N} (d_{j'} + d_{j''})$. 
	That is, $\capa(s)$ is equal to the sum of sizes of the buyers' minimal preferred bundles.
	The following lemma states that prices $\f{p}$ are competitive if and only if the value of a max flow in $G(\f{p})$ is equal to $\capa(s)$.
	
	\begin{lemma}\label{lem:competitive_flow}
		The prices $\f{p}$ are competitive if and only if there is a flow $f$ in $G(\f{p})$ of value
		$\val(f) = \capa(s)$.
		Moreover, given a competitive price vector $\f{p}$, an associated stable allocation $\f{x}$ can be computed via a single max flow computation in $G(\f{p})$.
	\end{lemma}
	\begin{proof}
		There is a clear one-to-one correspondence between an integral $s$-$t$-flow and a feasible assignment of objects. 
		Since all capacities in $G(\f{p})$ are integral, there exists a max flow with integral values. 
		By assigning $f_{j'i} + f_{j''i}$ items of object $i$ to buyer $j$ we obtain a feasible assignment.
		Note that at least one of the summands is zero since at most one of the arcs has positive capacity ($\Omega_j' \cap \Omega_j'' = \emptyset$).
		This assignment is competitive if and only if the demand of every buyer at prices $\f{p}$ is satisfied. 
		This is equivalent to the requirement that the flow satisfies all $s$-leaving arcs, i.e., $\val(f) = \sum_{j\in N} (d_{j'} + d_{j''}) = \capa(s)$.
	\end{proof}
	
	However, in this allocation $\f{x}$, no item with payoff zero is included. 
	Since we aim for an allocation where as much as possible is sold, we have to allocate the other items as well. 
	This is easy for those items which have price zero, since then every buyer with unsatisfied demand is willing to buy them.
	We show in \Cref{theorem:market-clearing} that there are enough buyers who are willing to buy the items with a positive price as well.
	
	We extend $G(\f{p}^*)$ to a flow network $H(\f{p}^*)$ such that the assignment of buyers to objects of payoff zero is possible.
	To do so, we first balance the supply and the demand. 
	If $\sum_{i \in \Omega} b_i < \sum_{j \in N} d_j$, we add a dummy object $i_0$ with supply $b_{i_0} = \sum_{j \in N} d_j - \sum_{i \in \Omega} b_i$ and valuations $v_{i_0j} = 0$ for all $j \in N$.
	If $\sum_{i \in \Omega} b_i > \sum_{j \in N} d_j$, we add a dummy buyer $j_0$ with demand $d_{j_0} = \sum_{i \in \Omega} b_i - \sum_{j \in N} d_j$ and valuations $v_{ij_0} = 0$ for all $i \in \Omega$.
	Now we can assume that $\sum_{i \in \Omega} b_i = \sum_{j \in N} d_j$.
	Note that the \ref{alg:price-raising} computes the same prices with the dummy object or buyer as without.
	To construct $H(\f{p}^*)$ from $G(\f{p}^*)$, additionally we add a vertex $j'''$ for each buyer $j \in N$ and an arc $(s,j''')$ with capacity $d_{j'''}$.
	Furthermore, we add for each $j \in N$ the arc $(j''',i)$ with capacity $b_i$ for $i \in \Omega_j'''$.
	
	\begin{proposition}\label{proposition:allocation}
		A max flow in network $H(\f{p})$ and its corresponding allocation satisfy:
		\begin{enumerate}[topsep=0.5ex, itemsep=0pt]
			\item A feasible flow in $G(\f{p})$ is a feasible flow in $H(\f{p})$.
			\item If for buyer $j \in N$ the flow on the arcs $(s, j')$ and $(s, j'')$ is saturated, $j$ is assigned to one of her preferred bundles at prices $\f{p}$.
			\item If $\f{p}$ is a Walrasian price vector, the allocation induced by a max flow is stable, i.e., every buyer obtains a preferred bundle.
		\end{enumerate}
	\end{proposition}
	Note that part 3 of the proposition is due to the fact that there exists an allocation where everything is sold (including dummy items), so the max flow saturates all $s$-leaving arcs, since demand and supply are balanced. 
	Hence in the corresponding allocation every buyer gets a preferred bundle.
	
    \vspace{2ex}
 
	\subsection{Computation of the buyer-optimal Walrasian equilibrium}
	Here we formally describe the \ref{alg:price-raising}.
	Each of its iterations can be done in polynomial time since the network can be constructed in polynomial time and only one max flow computation is needed.
	
	The intuition of the \ref{alg:price-raising} is to increase the price of a set of objects whenever the demand on this set exceeds the supply. 
	It is natural to increase the prices of the objects of an overdemanded subset until the buyers that were interested in these objects get interested in other objects as well. 
	This is exactly what happens in the following algorithm.
	\begin{algorithm}
		\SetAlgoRefName{Price-Raising Algorithm}
		\KwIn{Supplies $b_i$, demands $d_j$, valuations $v_{ij}$ for all $i \in \Omega$ and all $j \in N$, initial prices $\f{p}_0 \coloneqq \f{0}$}
		\KwOut{Buyer-optimal Walrasian prices $\f{p}^*$}
		Initialize: $\ell \coloneqq 0$\\
		Construct the auxiliary network $G_0 \coloneqq G(\f{p}_0)$ and find an integral max flow $f^0$ in it\\
		\While{$\val(f^\ell) < \capa_{G_\ell}(s)$}{
			Determine the left-most min cut $C_\ell$\\
            Increase prices on $I_\ell \coloneqq C_{\ell} \cap \Omega$, i.e., $\f{p}_{\ell+1} \coloneqq \f{p}_{\ell} + \f{\chi}_{I_\ell}$\\
			Construct the network $G_{\ell+1} \coloneqq G(\f{p}_{\ell +1 })$\\
			Find an integral max flow in $G_{\ell+1}$\\
			$\ell \coloneqq \ell +1$
		}
		\Return{$\f{p}^* \coloneqq \f{p}_{\ell}$}
		\caption{}
		\label{alg:price-raising}
	\end{algorithm}
	
	In \Cref{section:prices} we show that the prices computed by the \ref{alg:price-raising} are the component-wise minimum competitive prices.
	
	Given component-wise minimum competitive prices $\f{p}^*$, the \ref{alg:allocation} constructs the auxiliary network $H(\f{p}^*)$ and its max flow, leading to allocation $\f{x}^*$.
	Hence $\f{x}^*$ can be found in polynomial time. \Cref{section:allocation} shows that the value of the maximum flow is $\max\{\sum_{i \in \Omega} b_i, \sum_{j \in N} d_j\}$, since we include a dummy buyer resp.\ a dummy object.
	With \Cref{proposition:allocation} this implies that each buyer is assigned to one of her preferred bundles.
	
	\begin{algorithm}
		\SetAlgoRefName{Allocation Algorithm}
		\KwIn{Supplies $b_i$, demands $d_j$, $\Omega_j'(p^*)$, $\Omega_j''(p^*)$, $\Omega_j'''(p^*)$ for all buyers $j \in B$, prices $\f{p}^*$}
		\KwOut{Stable allocation $\f{x}^*$ where as much as possible is sold}
		Construct the auxiliary network $H(\f{p}^*)$\\
		Find an integral max flow $f$ in $H(\f{p}^*)$\\
		Determine the allocation $\f{x}^*$ given by the flow on the arcs $(j,i)$\\
		\Return{$\f{x}^*$}
		\caption{}
		\label{alg:allocation}
	\end{algorithm}
	
	If the supply does not exceed the demand, everything is sold. 
	Otherwise, only the items which are allocated to the dummy buyer $j_0$ are not sold. 
	Since $v_{ij_0}=0$ for all $i \in \Omega$, the price of an object allocated to $j_0$ has to be zero.
	Thus, all objects with positive price are completely sold.
    \pagebreak
    
	\begin{theorem}\label{thm:algo_computes_minimal_market-clearing_prices}
		The prices $\f{p}^*$ computed by the \ref{alg:price-raising} are the component-wise minimum competitive prices. 
		Moreover, under the stable allocation $\f{x}^*$ returned by the \ref{alg:allocation} as much as possible is sold, and every item with positive price is sold. 
		Thus, the prices $\f{p}^*$ coincide with the buyer-optimal Walrasian price vector.
	\end{theorem}
	
	We prove Theorem~\ref{thm:algo_computes_minimal_market-clearing_prices} in Section~\ref{sec:walrasian}.
	Note that the fact that the flow-based ascending auction returns the buyer-optimal Walrasian prices can also be shown directly by combining results of the literature on strong substitute valuation functions \cite{ausubel2006efficient} with some observation (see \Cref{sec:comparison}). 
	Ausubel \cite{ausubel2006efficient} showed that the buyer-optimal prices are computed by an ascending auction, where prices are always raised on the inclusion-wise minimal set corresponding to the steepest descent direction of the Lyapunov function. 
	We observe that these sets correspond to minimal maximum overdemanded sets which are exactly the sets we compute with a left-most min cut computation.
	However, the proof given in the subsequent sections has two advantages. 
	First, we show that the minimum Walrasian prices $\f{p}^*$ coincide with the minimum competitive prices, which will turn out to be crucial when proving our monotonicity results in Section~\ref{sec:monotonicity}. 
	Second, our proof is independent from the literature on strong substitute and \Lnat-convex functions and uses only network flow arguments.

    Regarding the running time, our considerations above yield the following lemma.
    \begin{lemma}
		One iteration of the \ref{alg:price-raising} requires at most $\abs{N}$ tier oracle calls, $\mathcal{O}(\abs{N}\abs{\Omega})$ time to construct the network and one max flow min cut computation.
	\end{lemma}

    For the max flow min cut computation, one can use, for example, the algorithm of Chen et al.~\cite{chen2022maximum} to find a max flow in $\mathcal{O}((\abs{N}\abs{\Omega})^{1+o(1)})$, and to find the left most min cut, one can use a breadth first search (BFS) from $s$ in the residual network in $\mathcal{O}(\abs{N} \abs{\Omega})$) time.
    This results in total running time of $\mathcal{O}(\abs{N}\TO + (\abs{N}\abs{\Omega})^{1+o(1)})$.

    Thus, our \ref{alg:price-raising} is faster than the running time of the algorithm for the general strong gross substitute case by Murota et al.~\cite{murota2013computing}, which is $\mathcal{O}(\abs{N} \DO + \abs{N} \abs{\Omega}^4 \log(\abs{N} \abs{\Omega} B) \ExO)$, where $B$ is the maximum number of copies of an object.
    Note that different oracle models are used in the two algorithms, but since the running time of the tier oracle can be bounded by $O(\abs{\Omega} \log(\abs{\Omega}))$ for a value oracle see \Cref{alg:preferred_bundles}, we consider this a fair statement. 
    In a follow up paper of this work~\cite{eickhoff2023faster} we improved the running time for the general case to $\mathcal{O}(\abs{N}\DO + \abs{N} \abs{\Omega}^3 \ExO)$.
    Again different oracle models are used and thus the running times are not directly comparable. As before, we can afford a running time of $O(\abs{\Omega} \log(\abs{\Omega}))$ for the tier oracle without exceeding the running time of the general mechanism.
    If we assume in favor of the general mechanism that the exchange and demand oracle queries takes $\mathcal{O}(1)$ time, our algorithm for the special case is still faster in all instances where $|N|^{o(1)} < |\Omega|^{2-o(1)}$, i.e., in all instances where the number of buyers is not drastically larger than the number of items.

	As shown in \cite{murota2013computing}, the number of iterations of the \ref{alg:price-raising} is at most $\left\Vert \f{p}^\ast - \f{p}_0\right\Vert_\infty$, where $\f{p}^\ast$ is the minimum Walrasian price vector. 
	Since $p^*(i) \leq \max_{j \in \Omega} v_{ij}$ this is pseudo-polynomial.
	Note, however, that we may as well start with any alternative start vector $\f{p}_0$, as long as $\f{p}_0$ is known to be component-wise smaller or equal to the unique minimum competitive price vector.
 
	In \Cref{sec:adapted_step_length} we present a variation of the auction where the price is raised as far as possible for the same inclusion-wise minimal maximum overdemanded set.
	
	\subsection{Auctioneer with memory -- Warm start with flow updates}\label{sec:flow_updates}

    A very natural idea is to not compute the flow completely from scratch in every iteration but instead to update the flow, when updating the price vector.
	We will describe this procedure and analyze the resulting running time.
	Interestingly, in one iteration this may not speed up the algorithm, but when combining this with the adapted step length algorithm described in \Cref{sec:adapted_step_length}, then the running time of the complete auction improves.
	
	To update the flow from one iteration to the next, assume we are given a maximum flow $f^\ell$ in the network induced by price vector $\f{p}_\ell$. 
	Moreover, let $C$ be the left most min cut in the network $G(\f{p}_\ell)$ and assume $C \neq \{s\}$ (because in this case the computed prices are already competitive and the auction terminates). 
	Then the \ref{alg:flow_update} shows how to update the flow to obtain a new feasible flow in $G(\f{p}_{\ell+1})$.
	The idea is to keep as much flow as possible from the previous assignment. 
	We use that items might change their set, i.e., go from  $\Omega_{j}'(\f{p}_\ell)$ to $\Omega_{j}''(\f{p}_{\ell +1})$ or vice versa, i.e., from $\Omega_{j}''(\f{p}_\ell)$ to $\Omega_{j}'(\f{p}_{\ell +1})$ but that for a fixed buyer this is possible in at most one direction.
	A third possibility is that a buyer might lose interest in an item, but in this case this item is not in $\Omega_{j}'(\f{p}_{\ell+1}) \cup \Omega_{j}''(\f{p}_{\ell+1})$ and thus no flow is assigned.
	
	\begin{algorithm}
		\SetAlgoRefName{Flow Update Algorithm}
		\KwIn{Network $G(\f{p}_\ell)$ with max flow $f^{\ell}$, $\f{p}_{\ell+1} \coloneqq \f{p}_\ell + \f{\chi}_I$, new network $G(\f{p}_{\ell+1})$}
		\KwOut{A feasible flow $f$ in $G(\f{p}_{\ell+1})$ with $D_{\f{p}_\ell} -\val(f^\ell) \geq D_{\f{p}_{\ell+1}} - \val(f)$}
		$f \coloneqq 0$ for all arcs in $G(\f{p}_{\ell+1})$\\
		\For{$(j,i)\in N \times \Omega$ with $f^\ell_{j'i}>0$ or $f^\ell_{j''i}>0$}{
			\uIf{$(j',i)$ is an arc in $G(\f{p}_{\ell+1})$, i.e., $i \in \Omega_j'(\f{p}_{\ell+1})$}{
				Add $\left(f^\ell_{j'i} + f^\ell_{j''i}\right)$ flow units to $f$ on the path $s$-$j'$-$i$-$t$
			}
			\ElseIf{$(j'',i)$ is an arc in $G(\f{p}_{\ell+1})$, i.e., $i \in \Omega_j''(\f{p}_{\ell+1})$}{
				Add $\left(f^\ell_{j'i} + f^\ell_{j''i}\right)$ flow units to $f$ on the path $s$-$j''$-$i$-$t$
			}
		}
		\Return{$f$}
		\caption{}
		\label{alg:flow_update}
	\end{algorithm}
	
	This flow update is well defined, since feasibility is shown in the following Lemma.
	\begin{lemma}\label{lemma:flow_update_feasible}
		The flow $f$ computed in \ref{alg:flow_update} is a feasible flow in $G(\f{p}_{\ell+1})$.
	\end{lemma}
	\begin{proof}
		Note that flow conservation is fulfilled in every node by construction. 
		It remains to show that the capacities in $G(\f{p}_{\ell+1})$ are obeyed.
		Let $C$ be the left-most min cut in iteration $\ell$ and $I=C \cap \Omega$ the set of items on which the prices are increased.
		\begin{itemize}
			\item For $(s,j')$ this is given: 
			we assign $\sum_{i \in \Omega_j'(\f{p}_{\ell+1})} \left(f^{\ell}_{j'i}+f^\ell_{j''i}\right)$ units of flow and the capacity is given by $\sum_{i \in \Omega_j'(\f{p}_{\ell+1})}b_i$. 
			Note that $\left(f^{\ell}_{j'i}+f^\ell_{j''i}\right) \leq b_i$ since the only $i$ leaving arc has capacity $b_i$.
			\item For $(s,j'')$ we consider how $d_{j''}(\f{p}_\ell)$ changes if we increase the prices on objects in $I$.
			\begin{itemize}
				\item If there are objects in $\Omega_j''(\f{p}_{\ell}) \cap \Omega_j'(\f{p}_{\ell+1})$, let $M \coloneqq \Omega_j''(\f{p}_{\ell}) \cap \Omega_j'(\f{p}_{\ell+1})$ be the objects which move from $\Omega_j''$ to $\Omega_j'$ and $S \coloneqq \Omega_j''(\f{p}_{\ell}) \cap \Omega_j''(\f{p}_{\ell+1})$ be the items that stay in $\Omega_j''$.
				
				We know that the objects in $\Omega_j''(\f{p}_{\ell}) \cap I \subseteq \Omega_j''(\f{p}_{\ell+1})$. 
				Moreover, the objects in $\Omega_j'(\f{p}_{\ell})$ stay in $\Omega_j'(\f{p}_{\ell+1})$, i.e.\ $\Omega_j'(\f{p}_{\ell}) \subseteq \Omega_j'(\f{p}_{\ell+1})$.
				Hence, the demand $d_{j'}$ increases and the demand $d_{j''}$ decreases by the number of items in $M$, i.e., $\sum_{i \in M} b_i$.
				
				Recall that left-most min cut $C$ is defined by the vertices reachable from $s$ in the residual network corresponding to $f^\ell$.
				By definition $M \cap I = \emptyset$ and $S \subseteq I$.
				
				If $j''$ is in the left-most min cut, the arcs in the cut which are leaving this node are fully used in any maximum flow. 
				Thus, $f^\ell_{j''i} = c_{j''i}=b_i$ for $i \in M$.
				That $c_{j''i} = \min \{ b_i, d_{j''} \} = b_i$ follows since otherwise with this item the complete demand could be satisfied. 
				But this is a contradiction since buyer $j$ wants to buy items with less payoff at prices $\f{p}_{\ell + 1}$, in other words $i \in \Omega_j'(\f{p}_{\ell+1})$ and not $i \in \Omega_j''(\f{p}_{\ell+1})$.
				
				Since this flow is shifted to $j'$ after the price update, the flow is reduced by the number of items in $M$ as well, thus the capacities are still obeyed.
				
				If $j''$ is not in the left-most min cut, we cannot reach $j''$ from items in $S \subseteq I$. 
				Hence, the flow on the edges $(j''i)$ with $i \in S$ is zero (otherwise the backwards arc exists in the residual network and $j''$ is reachable). 
				Hence, in the \ref{alg:flow_update} we do not assign any flow to an edge through the vertex $j''$. 
				Thus, the capacity of $(s,j'')$ is still not exceeded.
				
				\item Consider the case where $\Omega_j''(\f{p}_{\ell}) \cap \Omega_j'(\f{p}_{\ell+1}) = \emptyset$.
				
				If $\Omega_j''$ does not lose any objects by the price update, we know that
				\[
				d_j''(\f{p}_{\ell+1}) = d_j''(\f{p}_{\ell}) + \hspace{3ex} \sum_{\mathclap{\hspace{7ex} i \in \Omega_j'(\f{p}_{\ell}) \cap \Omega_j''(\f{p}_{\ell+1})}}\hspace{1ex} b_i \hspace{7ex} .
				\]
				Hence, the capacity is not exceeded on $(s,j'')$.
				
				It remains to show that the capacity is not exceeded if $\Omega_j''$ loses some objects (it might get new ones from $\Omega_j'$ as well). 
				The only situation which can cause problems is if the demand of objects in $\Omega''$ decreases.
				The demand $d_j''$ only decreases if there are objects moving from $\Omega_j''$ to $\Omega_j'$ (which does not happen by assumption) or if there are not enough objects with a positive payoff. 
				The latter case implies that all items available in $\Omega_j''(\f{p}_{\ell+1})$ are demanded and thus, the capacity constraint on $(s,j'')$ is fulfilled.
			\end{itemize}
			
			\item For $(i,t)$ the capacity does not change, so the capacity constraints are fulfilled.
			
			\item For $(j',i)$ and $(j'',i)$ it follows directly by the definition of the capacity and since the capacity on the $s$-leaving and $t$-entering arcs is not exceeded.\qedhere
		\end{itemize}
	\end{proof}

    \begin{lemma}
        \label{lem:running_time_adaption}
        The adaption of the network and the check whether the set of objects $I$ in the left-most min cut changes runs in time $\mathcal{O}(\abs{N}\TO + \abs{N}\abs{\Omega})$.
    \end{lemma}
    \begin{proof}
        After asking every buyer to report their preferences for the new prices, every arc is checked and then in constant time the flow is adapted (see \ref{alg:flow_update}). Afterwards, the left-most min cut can be computed by a breadth first search (BFS) in the residual network. As there are $\mathcal{O}(\abs{N}\abs{\Omega})$ many arcs, the statement follows.
    \end{proof}
	
	We can use the structure of the left-most min cut to show that if the flow decreases in an update step, than the demand decreases by at least the same amount. 
	This will help us to show that the prices returned by the auction are market-clearing (see \Cref{section:allocation}).
	Moreover, this enables us to give an upper bound on the overall running time of the ascending auction with adapted step length (see \Cref{sec:adapted_step_length}).
	
	\begin{lemma}\label{lemma:flow_demand}
		Given a max flow in $G(\f{p}_\ell)$ with corresponding left-most min cut $C$. 
		Let $I = C \cap \Omega$ be the overdemanded set and $\f{p}_{\ell+1}$ be the price vector after the price update. 
		Then, the demand $d_{j'} + d_{j''}$ of buyer $j$ will decrease at least by
		\begin{equation}
			\sum_{\mathclap{\substack{i \in \Omega_j''(\f{p}_{\ell}): \\ i \notin \Omega_j'(\f{p}_{\ell+1}) \cup \Omega_j''(\f{p}_{\ell+1})}}} \hspace{1ex} f^\ell_{j''i}.
			\label{removed_flow}
		\end{equation}
	\end{lemma}
	\begin{proof}
		If \eqref{removed_flow} is zero, the statement is directly fulfilled. 
		Thus, from now on, we consider the cases when flow is removed. 
		This can only occur if some objects get a payoff of zero after the price update, i.e., if
		\[
		S \coloneqq \Omega_j''(\f{p}_{\ell}) \setminus (\Omega_j'(\f{p}_{\ell+1}) \cup \Omega_j''(\f{p}_{\ell+1})) \neq \emptyset.
		\]
		Note that $S$ describes the set of objects which move from $\Omega_j''$ to $\Omega_j'''$. 
		Since $\Omega_j'''$ just contains objects with utility zero and $\Omega_j''$ only those with a positive payoff, all objects in $S$ are contained in $I$.
		
		By definition, the demand $d_{j'}+d_{j''}$ will decrease by
		\begin{align*}
			\max\Big\{0, d_{j''}(\f{p}_{\ell}) - \hspace{2ex} \sum_{\hspace{3ex} \mathclap{i \in \Omega_j''(\f{p}_{\ell}) \setminus I}} b_i \hspace{3ex}\Big\}.
		\end{align*}
		
		If $j''$ is contained in the left-most min cut, all arcs from $j''$ to $\Omega_j''(\f{p}_{\ell}) \setminus I$ are fully satisfied (since they are not reachable from $j''$).
		Hence, there can be at most $d_{j''}(\f{p}_{\ell}) - \sum_{i \in \Omega_j''(\f{p}_{\ell}) \setminus I} b_i \geq 0$ units of flow going through $j''$ to vertices in $S$. 
		Thus, for buyer $j$ we reduce the demand at least by the removed flow units traveling through a vertex of $j$.
		
		If $j'$ is not contained in the left-most min cut, there is no flow on the arcs from $j''$ to $I$. 
		Thus, no flow through buyer $j$ is removed and we are done in this case as well.
	\end{proof}
	
	\begin{corollary}\label{coro:demand-flow_value_non_decreasing}
		The flow $f$ computed in the \ref{alg:flow_update} satisfies
		\[
		D_{\f{p}_\ell}- \val(f^\ell) \geq D_{\f{p}_{\ell+1}} - \val(f). 
		\]
	\end{corollary}

	\section{Prices are Walrasian and buyer-optimal}\label{sec:walrasian}
	We will show that the \ref{alg:price-raising} returns buyer-optimal Walrasian prices by showing separately that the prices are the component-wise minimum competitive ones and that the prices are market clearing.

	\subsection{The auction returns component-wise minimum competitive prices}\label{section:prices}
	We now analyze the \ref{alg:price-raising} by considering the behavior of the buyers at given prices $\f{p}$ using the structure of $G(\f{p})$. 
	As usual, for a given digraph $(V, A)$, we use $\Gamma^-(v)$ to refer to all nodes that are the starting node of an incoming arc into $v$, i.e., $\Gamma^-(v) \coloneqq \{u \in V \mid (u, v) \in A\}$, analogously for $\Gamma^+(v)$. 
	We extend this definition also to sets, so that $\Gamma^-(I) \coloneqq \{u \in V \mid \text{there exists } v \in I \text{ with } (u, v) \in A\}$.
	
	For a given set of items $I \subseteq \Omega$, we denote the sum of demands of all buyers which cannot be fulfilled by items which are not in $I$ by
	\begin{equation*}
		d_{\f{p}}(I) \coloneqq \sum_{\jb \in \Gamma^-(I)} \max\Big\{0, d_{\jb} - \hspace{1ex} \sum_{\mathclap{\hspace{3ex} i \in \Gamma^+(\jb) \setminus I}} \hspace{1ex} c_{ij} \hspace{1ex} \Big\}.
	\end{equation*}
	Note that $d_{\f{p}}(I)$ is a natural lower bound on the number of copies from $I$ that are needed in every stable allocation.
	We call an item set $I$ \emph{overdemanded} if $\sum_{i \in I} b_i < d_{\f{p}}(I)$.
	The following condition \eqref{Hall-condition} is a consequence of Hall's Theorem for $b$-matching (cf. \cite{hall1987representatives}).
	
	\begin{lemma}
		Prices $\f{p}$ are competitive if and only if
		\begin{equation}\label{Hall-condition}
			\sum_{i \in I} b_i \geq d_{\f{p}}(I) \quad \text{for all } I \subseteq \Omega.
		\end{equation}
		Moreover, if prices $\f{p}$ are not competitive, and $C$ is the left-most min cut in $G(\f{p})$, then $C \cap \Omega$ is an overdemanded set. 
	\end{lemma}
	\begin{proof}
		Recall from Lemma~\ref{lem:competitive_flow} that prices $\f{p}$ are competitive if and only if there exists a flow $f$ in $G(\f{p})$ of value $\val(f) = \capa(\{s\}) = \sum_{j \in N} d_{j'} + d_{j''}$. 
		Thus, by the Max Flow Min Cut Theorem, prices $\f{p}$ are competitive if and only if $\{s\}$ is a min cut in $G(\f{p})$.
		
		First, we show that \eqref{Hall-condition} is necessary, i.e.\ that it holds if $\f{p}$ is competitive.
		Consider a flow $f$ in $G(\f{p})$ of value $\val(f) = \sum_{j \in N} d_{j'} + d_{j''}$ and some arbitrary set $I \subseteq \Omega$. 
		According to the flow conservation, for each vertex $\jb$ corresponding to a buyer $j \in N$ (i.e., either $j'$ or $j''$) we have
		
		\begin{align*}	
			\sum_{\mathclap{i \in \Gamma^+(\jb)\cap I}} \hspace{1ex} f_{\jb i} = \overbrace{f_{s\jb}}^{= d_{j'} \text{ resp. } d_{j''}} - \hspace{1ex} \sum_{\mathclap{\hspace{3ex} i \in \Gamma^+(\jb) \setminus I}} \hspace{1ex} \overbrace{f_{\jb i}}^{\leq c_{\jb i}} \hspace{1ex}
			\geq \max\Big\{0,\ d_{\jb} - \hspace{1ex} \sum_{\mathclap{ \hspace{3ex}i \in \Gamma^+(\jb) \setminus I}} \hspace{1ex} c_{\jb i} \hspace{1ex}\Big\}.
		\end{align*}
		If we sum over all $\jb \in \Gamma^-(I)$ we get
		\begin{align*}
			d_{\f{p}}(I)=\sum_{\mathclap{\jb \in \Gamma^-(I)}} \hspace{1ex} \max\Big\{0,\ d_{\jb} - \hspace{1ex} \sum_{\mathclap{\hspace{3ex} i \in \Gamma^+(\jb)\setminus I}} \hspace{1ex} c_{\jb i} \hspace{1ex} \Big\}
			\leq \sum_{\jb \in \Gamma^-(I)}  \hspace{4ex} \sum_{\mathclap{i \in \Gamma^+(\jb)\cap I }} f_{\jb i} 
			= \sum_{\mathclap{ i \in I}} \sum_{\jb \in \Gamma^-(I)} f_{\jb i} \hspace{1ex}
			= \sum_{i \in I} f_{it}
			\leq \sum_{i \in I} b_i.
		\end{align*}
		
		Thus, $I$ is not an overdemanded set. 
		Since we chose $I$ arbitrary, we have that \eqref{Hall-condition} is fulfilled.
		
		Now, we show that condition \eqref{Hall-condition} implies that prices $\f{p}$ are competitive.
		To see this, suppose that $\f{p}$ is not competitive, and let $C$ be the left-most min cut of $G(\f{p})$. 
		Since $\f{p}$ is not competitive, we know that 
		\begin{equation}
			\capa(C) < \capa(\{s\}) = \sum_{\jb \in N' \cup N''} d_{\jb}.\label{overdemanded_2}
		\end{equation}
		Define $I \coloneqq C \cap \Omega$ and $T \coloneqq C \cap (N' \cup N'')$. 
		The capacity of $C$ is given by
		\begin{equation}
			\capa(C) = \sum_{i \in I} b_i + \sum_{\jb \in T} \hspace{2ex} \sum_{\mathclap{\hspace{3ex} i \in \Gamma^+(\jb) \setminus I}} \hspace{1ex} c_{\jb i} \hspace{1ex} + \hspace{1ex} \sum_{\mathclap{\hspace{3ex} \jb \in (N' \cup N'') \setminus T}} \hspace{1ex} d_{\jb} \hspace{2ex}. \label{overdemanded_1}
		\end{equation}
		By combining and rearranging \eqref{overdemanded_1} and \eqref{overdemanded_2} we get the following chain of inequalities:
		\begin{equation*}
			\sum_{i \in I} b_i < \sum_{\jb \in T} d_{\jb} - \sum_{\jb \in T} \hspace{2ex} \sum_{\mathclap{\hspace{3ex} i \in \Gamma^+(\jb)\setminus I}} \hspace{1ex} c_{\jb i} \hspace{1ex} = \sum_{\jb \in T} \Big( d_{\jb} - \hspace{1ex} \sum_{\mathclap{\hspace{3ex} i \in \Gamma^+(\jb)\setminus I}} \hspace{1ex} c_{\jb i} \hspace{1ex} \Big) \leq \sum_{\jb \in T} \max\Big\{0, d_{\jb} - \hspace{1ex} \sum_{\mathclap{\hspace{3ex} i \in \Gamma^+(\jb) \setminus I}} \hspace{1ex} c_{\jb i} \hspace{1ex}\Big\} = d_{\f{p}}(I).
		\end{equation*}
		Hence, $I = C \cap \Omega$ is an overdemanded set and \eqref{Hall-condition} does not hold in this case.
	\end{proof}
	
	\begin{theorem}\label{thm:algo_computes_minimal_competitive_prices}
		The prices $\f{p}^*$ returned by our \ref{alg:price-raising} are the (unique) component-wise minimum competitive prices, i.e., if $\f{q}$ is a competitive price vector then $p^*(i) \leq q(i)$ for all $i \in \Omega$.
	\end{theorem}
	\begin{proof}
		Assume towards a contradiction that for a competitive price vector $\f{q}$, there is an object $i \in \Omega$ such that $p^*(i) > q(i)$. 
		Let $\f{p}_\tau$ denote the price vector in iteration $\tau$.
		Then, for the start vector $\f{p}_0$ we have $\f{p}_0 \leq \f{q}$. 
		Let $t$ be the last iteration where $\f{p}_t \leq \f{q}$, i.e., there is an object $i \in \Omega$ such that $p_{t+1}(i) > q(i)$. 
		Let $I$ be the overdemanded set chosen by the algorithm in iteration $t$ and split it into two parts $I^ = \cup I^<$ as follows:
		\begin{align*}
			I &= \{i \in \Omega \mid p_t(i) < p_{t+1}(i)\},
			& I^= &= \{i \in I \mid p_t(i) = q(i)\},
			& I^< &= \{i \in I \mid p_t(i) < q(i)\}.
		\end{align*}
		We will derive a contradiction by showing that $I^<$ induces a min cut $\tilde{C}$ which is a strict subset of $C$, since $I^=$ is non-empty by choice of $t$. 
		To do so, we start with analyzing the behavior of buyer $j \in N$ at prices $\f{q}$ by comparing it with the behavior at prices $\f{p}_t$. 
		For this purpose, we fix the network properties, i.e., we talk about everything w.r.t prices $\f{p}_t$ if not stated otherwise.
		In the following lemma, we will show that 
		\[
		\widetilde{C} = \{s\} \cup \Big\{\jb \in N' \cup N'' \mid d_\jb > \sum_{\substack{i \in \Gamma^+(\jb) \setminus I^<}} c_{\jb i}\Big\} \cup I^<
		\]
		is also a min cut with $\tilde{C} \subsetneq C$ (see Lemma~\ref{lem:smaller_mincut} below). 
		This, however, is a contradiction to $C$ being a left-most min cut, implying that the assumption that there is an object $i$ with price $p^*(i) > q(i)$ cannot be true. 
		Therefore, the \ref{alg:price-raising} finds the component-wise minimum competitive price vector.
	\end{proof}
	
	The following lemma is needed in the proof of     \Cref{thm:algo_computes_minimal_competitive_prices} and uses the notation which is described there.
    \pagebreak
        
	\begin{lemma}\label{lem:smaller_mincut}
		The cut $\widetilde{C} = \{s\} \cup \Big\{\jb \in N' \cup N'' \mid d_\jb > \sum_{\substack{i \in \Gamma^+(\jb) \setminus I^<}} c_{\jb i}\Big\} \cup I^<$ is a min cut and $\tilde{C} \subsetneq C$.
	\end{lemma}
	\begin{proof}
		Consider the buyers who need to buy items of $I$ such that $I$ becomes overdemanded.
		Let $T = C \cap (N' \cup N'')$, $T_1$ be the subset of buyers ($N' \cup N''$) demanding some objects in $I^=$, and $T_2$ are those who do not, i.e.,
		\begin{align*}
			T &= \Big\{\jb \in N' \cup N'' \mid d_\jb > \hspace{1ex} \sum_{\mathclap{\hspace{3ex }i \in \Gamma^+(\jb) \setminus I}} \hspace{1ex} c_{\jb i} \Big\},
			&T_1 &= \Big\{ \jb \in T \mid c_{\jb i} > 0 \text{ for an } i \in I^= \Big\},
			&T_2 &= T \setminus T_1.
		\end{align*}
		To see that this fits to the definition of $T$, we only need to check which vertices in $N' \cup N''$ are contained in the cut $C$. An illustration of the cut can be found in \Cref{fig:proof}.
		Recall that the capacity $\capa(C)$ of cut $C$ is defined as the sum of capacities on all outgoing arcs of $C$. 
		If a vertex $\jb$ is in the cut, $\capa(C)$ includes the sum over all capacities of arcs $(\jb, i)$ where $i \notin I$. 
		If vertex $\jb$ is not in the cut, $\capa(C)$ includes the capacity of the arc $(s, \jb)$ which is $d_\jb$. 
		Because $C$ is a left-most min cut, $C$ contains $\jb$ if and only if
		\[
		d_\jb > \hspace{1ex} \sum_{\mathclap{\hspace{3ex} i \in \Gamma^+(\jb) \setminus I}} \hspace{1ex} c_{\jb i}.
		\]
		
		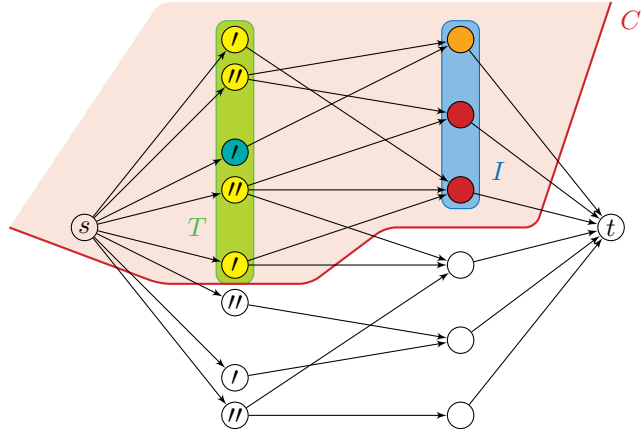
\begin{figure}[t]
			\centering
			\begin{tikzpicture}
				\path[rwthred,rounded corners,fill=rwthred!10] (8, 3) node[below right]{$C$} -- (7, 0) -- (5, 0) -- (4, -.75) -- (2, -.75) [sharp corners]-- (0, 0) [rounded corners]-- (2, 3) [sharp corners]-- cycle;
				
				\node[small vertex] (s) at (1, 0) {$s$};
				
				\draw[box=rwthgreen and rwthmay] (2.75, 2.75) rectangle (3.25, -.75); 
				\node[rwthgreen] at (2.5, 0) {$T$};
				\node[small vertex,fill=rwthyellow] (j11) at (3, 2.5) {$\prime$};
				\node[small vertex,fill=rwthyellow] (j12) at (3, 2) {$\prime\prime$};
				\node[small vertex,fill=rwthturkis] (j21) at (3, 1) {$\prime$};
				\node[small vertex,fill=rwthyellow] (j22) at (3, .5) {$\prime\prime$};
				\node[small vertex,fill=rwthyellow] (j31) at (3, -.5) {$\prime$};
				\node[small vertex] (j32) at (3, -1) {$\prime\prime$};
				\node[small vertex] (j41) at (3, -2) {$\prime$};
				\node[small vertex] (j42) at (3, -2.5) {$\prime\prime$};
				
				\draw[box=rwthblue and rwthlightblue] (5.75, 2.75) rectangle (6.25, .25); 
				\node[rwthblue] at (6.5, .75) {$I$};
				\node[small vertex,fill=rwthorange] (i1) at (6, 2.5) {};
				\node[small vertex,fill=rwthred] (i2) at (6, 1.5) {};
				\node[small vertex,fill=rwthred] (i3) at (6, .5) {};
				\node[small vertex] (i4) at (6, -.5) {};
				\node[small vertex] (i5) at (6, -1.5) {};
				\node[small vertex] (i6) at (6, -2.5) {};
				
				\node[small vertex] (t) at (8, 0) {$t$};
				
				\draw[rwthred,rounded corners,thick] (8, 3) -- (7, 0) -- (5, 0) -- (4, -.75) -- (2, -.75) -- (0, 0); 
				
				\foreach \j in {1, ..., 4} {
					\draw[->] (s) -- (j\j1);
					\draw[->] (s) -- (j\j2);
				}
				
				\draw[->] (j11) -- (i3);
				\draw[->] (j12) -- (i1);
				\draw[->] (j12) -- (i2);
				\draw[->] (j21) -- (i1);
				\draw[->] (j22) -- (i2);
				\draw[->] (j22) -- (i3);
				\draw[->] (j22) -- (i4);
				\draw[->] (j31) -- (i3);
				\draw[->] (j31) -- (i4);
				\draw[->] (j32) -- (i5);
				\draw[->] (j41) -- (i5);
				\draw[->] (j42) -- (i4);
				\draw[->] (j42) -- (i6);
				
				\foreach \i in {1, ..., 6} {
					\draw[->] (i\i) --(t);
				}
			\end{tikzpicture}
			\caption{Proof sketch. Illustration of a minimum $s$-$t$-cut $C$ and the induced sets $I$ (light blue) and $T$ (light green). The set $I^=$ is red, $I^<$ is orange, $T_1$ is yellow and $T_2$ is petrol.}
			\label{fig:proof}
		\end{figure}
		
		\begin{claim}
			It holds that
			\begin{equation}\label{I_1 not overdemanded}
				\sum_{i \in I^=} b_i \geq \sum_{\jb \in T_1} \min\Big\{d_\jb - \hspace{1ex} \sum_{\mathclap{\hspace{3ex} i \in \Gamma^+(\jb)\setminus I}} \hspace{1ex} c_{\jb i}\hspace{1ex}, \hspace{2ex} \sum_{\mathclap{\hspace{3ex}i \in \Gamma^+(\jb)\cap I^=}} \hspace{1ex} c_{\jb i} \hspace{1ex}\Big\}.
			\end{equation}
		\end{claim}
		\begin{subproof}
			First, consider the demand that $j'$ has on objects in $I^=$ at price $\f{p}_t$, i.e.,
			\[
			\sum_{i \in \Gamma^+(j') \cap I^=} c_{j'i}.
			\]
			Comparing $\f{q}$ and $\f{p}_t$, remember that the price of any item only increases, i.e., $p_t(i) \leq q(i)$, while the prices in $I^=$ remain the same. 
			That is why the buyer $j$ likes to buy at prices $\f{q}$ at least the same amount of items from objects in $I^=$ as at prices $\f{p}_t$.
			
			Next, we consider the demand $j''$ has on objects in $I$ at prices $\f{p}_t$. 
			Recall that $j''$ gets the same utility from every item in $\Omega_j''$. 
			Comparing $\f{p}$ and $\f{q}$ we have that the prices on $I^=$ remain constant while the prices on $I^<$ strictly increase. 
			Thus, $j$ likes to fill up the preferred bundle with items in $I^=$ and maybe with items outside of $I$ before buying the items in $I^<$. 
			Note however, that for prices $\f{q}$ it is not clear to which copy of buyer $j$ this demand is assigned. 
			Since furthermore, we have that buyer $j$ cannot buy more objects than available in $I^=$ we get the following lower bound on the demand reassigned from a buyer $j''$ for prices $\f{p}_t$ to buyer $j'$ or $j''$ for prices $\f{q}$:	
			\begin{equation*}
				\min\Big\{d_{j''} - \hspace{1ex} \sum_{\mathclap{\hspace{3ex} i \in \Gamma^+(j'')\setminus I}} \hspace{1ex} c_{j''i}\hspace{1ex} , \hspace{3ex} \sum_{\mathclap{\hspace{3ex} i \in \Gamma^+(j'')\cap I^=}} \hspace{1ex} c_{j''i}\Big\}.
			\end{equation*}
			Finally, if we sum up the demands of all $j' \in N'$ and $j'' \in N''$ at prices $q$ we obtain the following lower bound of the total demand of all buyers $j \in T_1$
			\begin{align*}
				&\sum_{j' \in T_1 \cap N'} \sum_{\substack{i \in \Gamma^+(j')\\ i \in I^=}} c_{j'i} +
				\sum_{j'' \in T_1 \cap N''} \min\Big\{d_{j''} - \sum_{\substack{i \in \Gamma^+(j'')\\ i \notin I}} c_{j''i}, \sum_{\substack{i \in \Gamma^+(j'')\\ i \in I^=}} c_{j''i}\Big\}\\
				\geq &\sum_{\jb \in T_1} \min\Big\{d_\jb - \sum_{\substack{i \in \Gamma^+(\jb)\\ i \notin I}} c_{\jb i}, \sum_{\substack{i \in \Gamma^+(\jb)\\ i \in I^=}} c_{\jb i}\Big\}.
			\end{align*}
			The set $I^=$ is not overdemanded at price $\f{q}$. 
			Thus, the demand of items in $I^=$ at price $\f{q}$ is smaller or equal than the total supply of all objects in $I^=$. 
			Using the bound of the total demand of all buyers in $T_1$, we obtain the inequality of the claim.
		\end{subproof}
		
		Next, we show that $\widetilde{C}$ is a min cut by comparing the capacities of $\widetilde{C}$ and of the left-most min cut $C$.
		\begin{align*}
			&\ \capa(C) - \capa(\widetilde{C})\\
			=&\ \Bigg(\sum_{i \in I} b_i + \sum_{\jb \in N' \cup N''} \hspace{-1ex} \min\Big\{d_\jb, \hspace{1ex} \sum_{\mathclap{\hspace{3ex} i \in \Gamma^+(\jb)\setminus I}}\hspace{1ex} c_{\jb i} \hspace{1ex}\Big\}\Bigg) - \Bigg(\sum_{i \in I^<} b_i + \sum_{\jb \in N' \cup N''} \hspace{-1ex} \min\Big\{d_\jb, \hspace{1ex} \sum_{\mathclap{\hspace{3ex} i \in \Gamma^+(\jb)\setminus I^<}} \hspace{1ex} c_{\jb i} \hspace{1ex} \Big\}\Bigg)\\
			=&\ \sum_{i \in I^=} b_i + \sum_{\jb \in N' \cup N''} \underbrace{\bigg(\min\Big\{d_\jb, \underbrace{\hspace{2ex} \sum_{\mathclap{\hspace{3ex} i \in \Gamma^+(\jb)\setminus I}} \hspace{1ex} c_{\jb i} \hspace{1ex}}_{\text{\small{$\eqqcolon \alpha$}}}\Big\} - \min\Big\{d_\jb, \underbrace{\hspace{2ex} \sum_{\mathclap{\hspace{3ex}i \in \Gamma^+(\jb)\setminus I_2}} \hspace{1ex} c_{\jb i} \hspace{1ex}}_{\text{\small{$\eqqcolon \beta$}}}\Big\}\bigg)}_{\text{\small{ $= 0$ for $\jb \notin T$, since $d_\jb \leq \alpha$ (by definition) and $\alpha \leq \beta$}}}\\
			=&\ \sum_{i \in I^=} b_i + \sum_{\jb \in T} \Bigg(\hspace{2ex}\sum_{\mathclap{\hspace{3ex} i \in \Gamma^+(\jb)\setminus I}} \hspace{1ex} c_{\jb i} \hspace{1ex} - \hspace{1ex} \min\Big\{d_\jb, \hspace{1ex} \sum_{\mathclap{\hspace{3ex}i \in \Gamma^+(\jb)\setminus I^<}} \hspace{1ex} c_{\jb i} \hspace{1ex} \Big\}\Bigg)\\
			=&\ \sum_{i \in I^=} b_i - \sum_{\jb \in T}\min\Big\{d_\jb - \hspace{1ex} \sum_{\mathclap{\hspace{3ex} i \in \Gamma^+(\jb)\setminus I}} \hspace{1ex} c_{\jb i} \hspace{1ex}, \hspace{2ex} \sum_{\mathclap{\hspace{3ex} i \in \Gamma^+(\jb)\cap I^=}} \hspace{1ex} c_{\jb i} \hspace{1ex}\Big\} \overset{\eqref{I_1 not overdemanded}}{\geq} 0.	
		\end{align*}
		Therefore $\widetilde{C}$ is a min cut too.
		
		It remains to show that $\tilde{C} \subsetneq C$. 
		This follows directly by the definitions: 
		$\tilde{C} \cap (N' \cup N'') \subseteq T = C \cap (N' \cup N'')$ and $\tilde{C} \cap \Omega = I^< \subsetneq I = C \cap \Omega$, since $I^=$ is not empty.
	\end{proof}
	
	\begin{observation}
		The proof of Theorem~\ref{thm:algo_computes_minimal_competitive_prices} does not use that the \ref{alg:price-raising} starts at prices zero. 
		As long as $p_0(i) \leq p^*(i)$ for all $i \in \Omega$, the proof works. 
		Thus, the algorithm finds the minimum competitive price vector $\f{p}^*$ if it starts at prices $\f{p}_0 \leq \f{p}^*$.
	\end{observation}

    \pagebreak
    
	\subsection{The auction returns market-clearing prices}\label{section:allocation}
	To complete the proof of Theorem~\ref{thm:algo_computes_minimal_market-clearing_prices} it remains to show that the \ref{alg:price-raising} computes prices that are market-clearing, i.e., that an allocation exists where $D = \min \{\sum_{i \in \Omega} b_i, \sum_{j \in N} d_j\}$ is sold.
	
	\begin{theorem}\label{theorem:market-clearing}
		Given prices $\f{p}^*$ computed by the \ref{alg:price-raising}, there exists an allocation $\f{x}^* \in \mathbb{Z}_+^{\Omega \times N}$ such that:
		\begin{enumerate}
			\item Any buyer $j \in N$ gets a preferred bundle, i.e., $\f{x}^*_{\bullet j} \in D_j(\f{p}^*)$. \label{property:pref_bundle}
			\item If there is an item which is not sold, it has price zero, i.e., $\sum_{j \in N} x^*_{ij} = b_i$ for all $i \in \Omega$ with $p^*(i) > 0$.
			\item As much as possible is sold, i.e., $\sum_{i \in \Omega} \sum_{j \in N} x^*_{ij} = D$.
		\end{enumerate}
	\end{theorem}
	To show this theorem, we will use the following notation. 
	An allocation $\bar{\f{x}}$ is an \emph{extension} of $\f{x}$ if $\bar{x}_{ij} \geq x_{ij}$ for all $j \in N$ and $i \in \Omega$. 
	A flow $f$ in $G(\f{p})$ \emph{induces} an allocation $\f{x}^f$ by defining $x_{ij} = f_{ij'} + f_{ij'}$. 
	Note that an allocation induced by a flow in $G(\f{p})$ always distributes subsets of a preferred bundle w.r.t.\ prices $\f{p}$.
	\begin{proof}
		We construct an allocation with the desired properties by reusing the flow obtained during the execution of the \ref{alg:price-raising} which uses the \ref{alg:flow_update} as described in \Cref{sec:flow_updates}.
		The constructed allocation $\f{x}^*$ will be induced by the max flow in $G(\f{p}^*)$, thus property \ref{property:pref_bundle} is fulfilled by definition.
		
		We extend $\f{x}^*$ to an allocation where every item with a positive price is sold. 
		For $i \in \Omega$ with $\sum_{j \in N} x_{ij}^* < b_i$ we assign the remaining items to buyers $j$ who demanded items of this objects in the last iteration where the price on $i$ was increased, i.e., in the last iteration with $i \in I$. 
		More precisely, we assign it to those buyers for which the flow was removed. 
		Since the price is not increased in later iterations, the object is still in $\Omega_j'''(\f{p}^*)$.
		By \Cref{lemma:flow_demand} we know that the removed flow did not exceed $d_j - d_{j'}(\f{p}^*) - d_{j''}(\f{p}^*)$, since $d_{j'}+d_{j''}$ just decreases during the algorithm.
		Hence, the demand is not exceeded if we assign items according to the removed flow.
		
		This means all items with a positive price are assigned.
		This is due to the fact that objects in $I$ are fully assigned, since $I$ is given by a left-most min cut. 
		The flow on $(b_i, t)$ remains if we do not remove flow through $i$ which can only happen if $i \in I$. 
		Thus, all items $i$ with a positive price are assigned, either induced by the max flow or by the removed-flow in the last iteration with $i \in I$.
		
		It remains to extend this assignment to an allocation where as much as possible is sold.
		All unsold items have price $0$ and all players with left over demand have payoff $0$ on all these items. 
		This means we get a complete bipartite graph in which every maximum flow assigns either all remaining objects or satisfies all demands. 
		Since all prices and payoffs are 0, any such induced allocation can just be added to the computed assignment.
	\end{proof}
	
	Note that the allocation can also be computed by the \ref{alg:allocation}. 
	We mainly use the described procedure to prove that there is an allocation with the desired properties. 
	This answers the question why the algorithm returns market-clearing prices.
	
	\begin{proof}[Proof of \Cref{thm:algo_computes_minimal_market-clearing_prices}]
		The prices are the minimum competitive prices by \Cref{thm:algo_computes_minimal_competitive_prices}. 
		By \Cref{theorem:market-clearing} there is an associated stable allocation where $D$ is sold and where each item with positive price is sold. 
		This existence together with the construction of the network $H(\f{p}^*)$ gives us that the \ref{alg:allocation} computes such an allocation. 
		The existence of an allocation where as much as possible is sold also implies that the minimum competitive prices computed by the \ref{alg:price-raising} are buyer-optimal Walrasian prices.
	\end{proof}

	\section{Monotonicity}\label{sec:monotonicity}
	The \ref{alg:price-raising} determines for each instance
    the unique component-wise minimum Walrasian price vector $\f{p}^*$.
	In this section, we analyze the monotonicity of the auction with respect to changes of the demand and supply.
    This research question belongs to the field of comparative statics, which is commonly used in economics to analyze markets.
	In particular, we show that, as intuitively expected, the auction is monotone in the sense that the returned prices can only increase if the demand increases or the supply decreases. 
	
	\begin{theorem}\label{thm:monotonicity}
		Given an instance with valuations $v$, demands $\f{d}$, supplies $\f{b}$ and the corresponding buyer-optimal Walrasian prices $\f{p}$, consider a second instance with the same valuation functions but increased demands and decreased supplies, i.e., demands $\dnewb$ with $0 \leq d_j \leq \dnew_j$ for all $j \in N$ and supplies~$\bnewb$ with $0 \leq \bnew_i \leq b_i$ for all $i \in \Omega$ with buyer-optimal Walrasian prices $\pnewb$. 
		Then we have $p(i) \leq \pnew(i)$ for all $i \in \Omega$ with $\bnew_i > 0$.
	\end{theorem}
	
	The proof idea is that prices remain competitive when the supply increases or the demand decreases. Since for truncated additive valuation functions minimum competitive prices are unique and market-clearing, the buyer-optimal Walrasian prices are smaller or equal to any competitive prices.
	
	To prove the theorem we show two lemmas, analyzing the change of the demand and the supply separately.
	
	\begin{lemma}\label{lem:monotonicity_demand}
		Let $\f{d}$ and $\dnewb$ be two demand vectors with $d_j \leq \dnew_j$ for all $j \in N$ (here it is possible that $d_j = 0$ for some $j \in N$). 
		Then the buyer-optimal Walrasian prices $\f{p}$ at demand $\f{d}$ are not greater than the buyer-optimal Walrasian prices $\pnewb$ at demand $\dnewb$, i.e., $p(i)\leq \pnew(i)$ for all $i \in \Omega$.
	\end{lemma}
	\begin{proof}
		Consider an integral max flow $f'$ at prices $\pnewb$ in the auxiliary flow network $G(\pnewb)$.
		Recall that flow $f'$ corresponds to a feasible allocation from buyers to their preferred bundles. 
		Now adapt this flow as follows. 
		For each buyer $j$ with $d_j < \dnew_j$, among the paths going through $j'$ or $j''$, select one with a currently lowest payoff and reduce flow on that path, until flow through $j'$ and $j''$ gets reduced to $d_j$. 
		This procedure terminates with a flow meeting demands $d_j$. 
		Furthermore, since we still allocate the items with the highest payoff to a buyer, each buyer is allocated to a preferred bundle at prices $\pnewb$ at demand $\f{d}$. 
		Thus, $\pnewb$ is a competitive price vector for demand $\f{d}$. 
		By Theorem~\ref{thm:algo_computes_minimal_competitive_prices} it follows that $\f{p}$ is not only the minimum Walrasian price vector but also the component-wise minimum competitive one. 
		Thus, we get $p(i) \leq \pnew(i)$ for all $i \in \Omega$.
	\end{proof}
	
	Next we show that the minimum competitive prices are bigger if the supply is smaller.
	\begin{lemma}\label{lem:monotonicity_supply}
		Let $\f{b}$ and $\bnewb$ be two supply vectors with $\bnew_i \leq b_i$ for all $i \in \Omega$ (here it is possible that $\bnew_i=0$ for some $i \in \Omega$). 
		Then for the corresponding buyer-optimal Walrasian prices $\f{p}$ and $\pnewb$, it holds that $p(i) \leq \pnew(i)$ for all $i$ with $\bnew_i > 0$.
	\end{lemma}
	\begin{proof}
		Assume without loss of generality that $\f{b}$ and $\bnewb$ only differ in the supply of object $\ell$ by one item. 
		We fix the allocation computed by the minimum competitive prices at supply $\bnewb$ and we consider two cases.
		
		First, we consider the case $\bnew_\ell > 0$.
		Given the assigned bundles, we analyze the behavior of the buyers when the additional item of object $\ell$ arrives.
		If there is a buyer $j$ who is not assigned to one of her preferred bundles at supply $\f{b}$, we knew that she is the only one who is assigned to items of object $\ell$ since otherwise the prices $\pnewb$ are not competitive.
		Thus, all other buyers are assigned to one of their preferred bundles at supply $\f{b}$.
		Hence we can change the preferred bundle of buyer $j$ by assigning one more item of $\ell$ to her, if necessary by omitting the least profitable item.
		This change does not harm the other buyers, thus in the new allocation everyone is assigned to a preferred bundle at prices $\pnewb$.
		Thus, prices $\pnewb$ are competitive for the instance with supply $\f{b}$. 
		The prices $\f{p}$ are the minimum competitive prices by Theorem~\ref{thm:algo_computes_minimal_competitive_prices}, which yields $p(i) \leq \pnew(i)$.
		
		Next, we consider the case $\bnew_\ell=0$. 
		We adapt the prices $\pnewb$ to prices $\bar{\f{p}}$ by setting $\bar{p}(i) = \pnew(i)$ for $i \in \Omega \setminus \{\ell\}$ and $\bar{p}(\ell) = \max_{j \in B} v_{\ell j} +1$. 
		Thus, no buyer wants to buy an item of object $\ell$. 
		Therefore, the given allocation is an assignment of buyers to one of their preferred bundles at prices $\bar{\f{p}}$ for supply $\f{b}$. 
		Using again that $\f{p}$ is the minimum competitive price vector at supply $\f{b}$, we get $p(i) \leq \bar{p}(i) = \pnew(i)$ for all $i \in \Omega \setminus \{\ell\}$.
	\end{proof}
	
	\begin{proof}[Proof of \Cref{thm:monotonicity}]
		For a given modified instance we can construct an intermediate instance where only the demand is changed and the supply remains as in the original instance. 
		With \Cref{lem:monotonicity_demand} this implies that prices only increase compared to the original instance. 
		Now applying \Cref{lem:monotonicity_supply} to the intermediate instance gives the statement of the theorem.
	\end{proof}
	
	The monotonicity allows for faster reoptimization by starting with the old Walrasian price vector.
	\begin{corollary}
		Given prices $\f{p}$ we can compute $\pnewb$ by applying at most $\|{\f{p}-\pnewb}\|_{\infty}$ iterations of the \ref{alg:price-raising} with start prices $p(i)$ for all $i \in \Omega$ with $\bnew_i > 0$. 
	\end{corollary}
	
	\Cref{thm:algo_computes_minimal_market-clearing_prices} allows starting with any initial price vector which is in every component at most as large as the minimum competitive price vector.
	Thus, Theorem~\ref{thm:monotonicity} allows us to start the \ref{alg:price-raising} at the price $p(i)$ for all $i \in \Omega$ with $\bnew_i>0$. 
	Murota et al.~\cite{murota2013computing} show that the number of iterations is then bounded by $\max\{p(i) - \pnew(i)\mid i \in \Omega,\ \bnew_i > 0\}$.
	However, the following example shows that we cannot bound $\|{\f{p}-\pnewb}\|_{\infty}$ even if the demand or supply is only slightly changed:
	
	\begin{example}
		Consider an instance with two buyers and two objects. 
		The valuation of both buyers are $(M, M)$, the demand for both is two and the supply of both objects is two. 
		In this instance $\f{p} = (0, 0)$ are buyer-optimal competitive prices. 
		Now assume the demand of one buyer is increased by one. 
		In this case, the unique buyer-optimal prices are $\pnewb = (M, M)$ and thus $\|{\f{p}-\pnewb}\|_{\infty} = M$.
		The same happens if the supply of one item is decreased by one.
	\end{example}
	
	It is worth pointing out that monotonicity is not guaranteed if the valuation functions are not gross substitute. 
	\begin{example}
		Consider an instance with two buyers $N = \{1, 2\}$ and two objects $\Omega = \{\alpha, \beta\}$.
		The supply of each object is $1$, the demand of both buyers is unbounded.
		The valuations are not given element-wise but as a set function by
		\[
		v_i(\emptyset) = v_i(\{\alpha\}) = v_i(\{\beta\}) = 0, \quad v_i(\Omega) = 1
		\]
		for $i \in N$.
		The valuations are \emph{complementarities}, i.e., here buyers want to buy items only in a bundle with the other object otherwise they have no value to the buyers (think of left-hand gloves and right-hand gloves). 
		Gross substitutes always have the no-complementarities condition (see \cite{gul1999walrasian}).
		Any price vector that yields a total price of $1$, together with an allocation that gives both objects to the same buyer is at equilibrium.
		However, if we reduce the supply of either object to $0$, the unique Walrasian price vector is $(0, 0)$ as each object on its own is worthless to both buyers.
	\end{example}

	\section{Embedding into existing literature}\label{sec:comparison}
	The goal of this section is to show the connection between our algorithm and the existing work. 
	In the first part we consider the connection to ascending auctions for strong gross substitute valuations.
	It turns out that computing a left-most min cut in our network directly corresponds to finding the inclusion-wise minimal set defining a steepest descent direction of the Lyapunov function in case of truncated additive valuation functions. 
	In the second part we address the question whether the computed prices can be considered as VCG prices (see below for definition). 
	While this does not hold for gross substitute valuations it is shown for single-unit demands. Unfortunately, for truncated additive valuations no mechanism based on prices per item can determine VCG prices as we show later on.

	\subsection{Comparison to ascending auction}
	Note that in our model the Lyapunov function can be rewritten to
	\[
	L(\f{p}) = \max \sum_{j \in N} \sum_{i \in \Omega} (v_{ij} - p(i)) x_{ij} + \sum_{i \in \Omega} b_i \cdot p(i)
	\]
	subject to $\sum_{i \in \Omega} x_{ij} \leq d_j$ and $x_{ij} \in [0, b_i]$ for all $j \in N$ and all $i \in \Omega$.
	
	\begin{proposition}\label{prop:difference_lyapunovfunction}
		In our model with truncated additive valuation functions the difference of the Lyapunov function in an augmentation step equals the difference between the capacity of the $s$-leaving arcs and the min cut value. 
		More formally
		\[
		L(\f{p}) - L(\f{p} + \f{\chi}_X)=\hspace{2ex}\sum_{\mathclap{\jb \in \Gamma^-(X)}} \hspace{1ex}\max \Big\{0, d_\jb - \hspace{1ex}\sum_{\mathclap{ \hspace{3ex} i \in \Gamma^+(\jb)\setminus X}} \hspace{1ex} c_{\jb i} \hspace{1ex} \Big\} - \sum_{i \in X} b_i.
		\]
	\end{proposition}
	
	The first sum is the sum over the differences $V_j(\f{p}) - V_j(\f{p} + \f{\chi}_X)$ for each buyer.
	The difference corresponds to the total amount of items that $j$ wants to buy from $X$ under prices $\f{p}$ without any alternative outside of $X$. 
	This yields the equation by definition of the Lyapunov function.
	
	\begin{lemma}\label{lemma:comparison_murota}
		Given prices $\f{p}$, the overdemanded set $I$ determined by the left-most min cut in $G(\f{p})$ minimizes $L(\f{p} + \f{\chi}_X)$ among all $X \subseteq \Omega$.
	\end{lemma}
	\begin{proof}
		A set $X \subseteq \Omega$ minimizes $L(\f{p} + \f{\chi}_X)$ if and only if it maximizes 
		\[
		L(\f{p}) - L(\f{p} + \f{\chi}_X) = \sum_{j \in N} (V_j(\f{p}) - V_j(\f{p} + \f{\chi}_X)) - \sum_{i \in X} b_i.
		\]
		Now we consider again the constructed auxiliary networks $G(\f{p})$ and $G(\f{p} + \f{\chi}_X)$ and the induced changes in capacities. 
		With Proposition~\ref{prop:difference_lyapunovfunction}, we obtain that $X \subseteq \Omega$ maximizes
		\begin{equation}\label{eq:lyapunov_network}	
			\sum_{\jb \in \Gamma^-(X)} \max\{0, d_\jb - \hspace{1ex} \sum_{\mathclap{\hspace{3ex} i \in \Gamma^+(\jb)\setminus X}} \hspace{1ex} c_{\jb i} \hspace{1ex}\} - \sum_{i \in X} b_i
		\end{equation}	
		Given an $X$ that minimizes $L(\f{p} + \f{\chi}_X)$, we construct the cut
		\[
		C_X = \{s\} \cup \{\jb \in N' \cup N'' \mid d_\jb > \hspace{1ex} \sum_{\mathclap{\hspace{3ex} i \in \Gamma^+(\jb)\setminus X}} \hspace{1ex} c_{\jb i} \hspace{1ex} \} \cup X.
		\]
		The structure of the cut, determined by our algorithm from set $I$, is given by
		\[
		C = \{s\} \cup \{\jb \in N' \cup N'' \mid d_\jb > \hspace{1ex} \sum_{\mathclap{\hspace{3ex} i \in \Gamma^+(\jb)\setminus I}} \hspace{1ex} c_{\jb i} \hspace{1ex} \} \cup I.
		\]
		We can show that $C_X$ is also a min cut:
		\begin{align*}
			0 &\geq \capa(C) - \capa(C_X)\\
			&= \sum_{i \in I} b_i + \sum_{\jb \in N' \cup N''} \hspace{-1ex} \min\{d_\jb, \hspace{2ex} \sum_{\mathclap{\hspace{3ex} i \in \Gamma^+(\jb) \setminus I}} \hspace{1ex} c_{\jb i} \hspace{1ex} \} - \sum_{i \in X} b_i - \sum_{\jb \in N' \cup N''} \hspace{-1ex} \min\{d_\jb, \hspace{2ex} \sum_{\mathclap{\hspace{3ex}i \in \Gamma^+(\jb) \setminus X}} \hspace{1ex} c_{\jb i} \hspace{1ex}\}\\
			&= \sum_{i \in I} b_i + \sum_{\jb \in N' \cup N''} \bigg(d_\jb - \max\{0, d_\jb - \hspace{1ex} \sum_{\mathclap{\hspace{3ex} i \in \Gamma^+(\jb) \setminus I}} \hspace{1ex} c_{\jb i} \hspace{1ex}\} \bigg) \\
			&\phantom{=} - \sum_{i \in X} b_i - \sum_{\jb \in N' \cup N''} \bigg(d_\jb - \max\{0, d_\jb - \hspace{1ex} \sum_{\mathclap{\hspace{3ex}i \in \Gamma^+(\jb) \setminus X}} \hspace{1ex} c_{\jb i} \hspace{1ex}\} \bigg)\\
			&= \bigg(\sum_{\jb \in N' \cup N''} \hspace{-1ex} \max\{0, d_\jb - \hspace{1ex} \sum_{\mathclap{\hspace{3ex} i \in \Gamma^+(\jb)\setminus X}} \hspace{1ex} c_{\jb i} \hspace{1ex}\} - \sum_{i \in X} b_i \bigg) - \bigg(\sum_{\jb \in N' \cup N''} \hspace{-1ex} \max\{0, d_\jb - \hspace{1ex} \sum_{\mathclap{\hspace{3ex} i \in \Gamma^+(\jb) \setminus I}} \hspace{1ex} c_{\jb i} \hspace{1ex}\} - \sum_{i \in I} b_i\bigg)\\
			&\geq 0,
		\end{align*}
		where the first inequality follows from the fact that $C$ is a min cut and the last inequality follows from~\eqref{eq:lyapunov_network}, i.e., that $X$ maximizes the term. 
		Hence, $C_X$ is a min cut. 
		Moreover the choice $X = C \cap \Omega$ minimizes $L(\f{p} + \f{\chi}_X)$, since equality holds in the chain of inequalities. 
	\end{proof}

    \subsection{Adapted step length}\label{sec:adapted_step_length}
    One natural approach to speed up the computation of buyer-optimal prices is to increase the length of the augmentation steps. 
	In the \ref{alg:price-raising} this means given a left-most min cut, increase the prices of the corresponding objects until the min cut changes. 
	In the more general case of strong substitute valuation functions, this means a steepest descent direction of the Lyapunov function is used as long as it remains the steepest descent direction (see \cite{shioura2017algorithms} Section~3, Theorem~4.17).
	The fact that increasing the prices on the objects in a left-most min cut as long as possible is indeed a special case of the algorithm of Shioura follows by Lemma~\ref{lemma:comparison_murota}.
	
    How to find this step length is not immediate. It is not enough to ask all buyers for the information when the preferred sets $\Omega'$, $\Omega''$ and $\Omega'''$ change. It is well possible that these sets change but the items on which the prices should be increased may not.
    One possibility to find the correct step length is using a binary search. This comes with the drawback that it destroys the natural process of increasing prices but with the advantage that it allows us to use the same oracle model as before.
 
    That a binary search can be used for determining the step length is possible by the following proposition of Shioura \cite{shioura2017algorithms}.
	
	\begin{proposition}[\cite{shioura2017algorithms} Proposition 4.16]\label{prop:shioura_monotonie}
		In an ascending auction, whenever a  steepest descent direction of the Lyapunov function becomes infeasible one of the following two things happens: 
		either the support of the direction increases, or the slope with which the Lyapunov function changes decreases.
	\end{proposition}
	
	For us, this means whenever the set determined by a cut is not the inclusion-wise minimal maximum overdemanded set any longer, then either the number of objects in the left-most min cut increases or the overdemandedness decreases.
	Due to the resulting monotonicity we can apply binary search to determine the step length.
	
	This observation allows to bound the number of price raising steps as done by \cite{shioura2017algorithms} in a similar way for the general setting.
	
	\begin{algorithm}
		\SetAlgoRefName{Adapted Step Length Algorithm}
		\KwIn{Buyer with valuations and demand, objects with supply}
		\KwOut{Buyer optimal Walrasian prices}
		$\f{p} \coloneqq \f{0}$\\
		$f$ max flow in network $G(\f{p})$\\
		$C$ left most min cut in $G(\f{p})$\\
		\While{$\val(f) < D_{\f{p}}$}{
			$\alpha_{up} \coloneqq v_{\max}$\\
			$\alpha_{lo} \coloneqq 1$ \\
			$\alpha \coloneqq \lfloor \tfrac{\alpha_{up} + \alpha_{lo}}{2}\rfloor$\\
			\While{$\alpha_{up} - \alpha_{lo} > 0$}{
				Adapt network and flow to price $\f{p} + \alpha \cdot \f{\chi}_{C \cap \Omega}$\\
				Use breadth first search (BFS) in residual network\\
				\If{min cut is equal to C}{
					$\alpha_{lo} \coloneqq \alpha$
				}
				$\alpha_{up} \coloneqq \alpha$
			}
			Set $\f{p} \coloneqq \f{p} + \alpha \cdot \f{\chi}_{C \cap \Omega}$\\
			Adapt network and flow to $\f{p}$\\
			Augment flow $f$ as long as possible in $G(\f{p})$\\ (for each augmentation step we use BFS in residual network and augment along the path found)\\
			$C$ new left most min cut (this is all nodes reachable by BFS in residual network)
		}
		\Return{$\f{p}$}
		\caption{}
		\label{alg:adapted_step_length}
	\end{algorithm}
    \vspace{-2ex}
	
	\begin{lemma}
		The running time of an ascending flow auction with adapted step length as described in \ref{alg:adapted_step_length} is given by $\mathcal{O}(\log(v_{\max}) \abs{\Omega} D_{\f{0}} \abs{N} (\abs{\Omega} + \TO))$, where $v_{\max} = \max_{j \in N, i \in \Omega} v_{ij}$ and $D_{\f{0}} = \sum_{j \in N} d_j$.
	\end{lemma}
	\begin{proof}
		To prove the statement we will show two claims from which it follows directly.
		We start by bounding the time of an update step. 
		As an update step, we summarize a flow augmentation (by a value of one), a min cut computation, and an update of the network and the flow from one iteration to the next. 
		In the next step we will bound the total number of update steps needed. 
		Note that the number of update steps does not correspond to the iterations of the \texttt{while}-loop in \ref{alg:adapted_step_length}.
		
		\begin{claim}
			Every update step takes time at most $\mathcal{O}(\abs{N} \TO + \abs{N}\abs{\Omega})$.
		\end{claim}
		\begin{subproof}
           To increase the flow by one unit (or more), we use a breadth first search (BFS) in the residual network. 
			The same is true for the computation of a min cut.
            The BFS runs in time $\mathcal{O}(\abs{N}\abs{\Omega})$.

            By \Cref{lem:running_time_adaption} the adaption of the network runs in time $\mathcal{O}(\abs{N}\TO + \abs{N}\abs{\Omega})$.
		\end{subproof}
		
		Now, we bound the number of update steps.
		
		\begin{claim}
			The number of update steps is bounded by $\mathcal{O}(\log(v_{\max}) \abs{\Omega} D_{\f{0}})$.
		\end{claim}
		\begin{subproof}
			In total the auction finds some prices $p^*$ and some flow $f^*$ such that $\val(f^*) = D_{\f{p^*}}$ or in other words $\val(f^*) - D_{\f{p^*}} = 0$.
			It starts with the all $0$-flow and $D_{\f{0}}$.
			During the algorithm  $\val(f)-D_{\f{p}}$ is non-increasing (see \Cref{coro:demand-flow_value_non_decreasing}). 
			If it decreases in an update step it decreases by at least by one. 
			
			Using \Cref{prop:shioura_monotonie} we can show that this value decreases all $\mathcal{O}(\log(v_{\max}) \abs{\Omega})$ update steps.
			
			Given some network and some flow, if we can augment the flow, clearly the value $\val(f)-D_{\f{p}}$ decreases.
			So assume it is not possible to augment the flow. 
			With one update step we can recognize this situation and compute a min cut giving an inclusion-wise minimal maximum overdemanded set. 
			It is possible that we adapt the network in the next update step, but in the resulting graph, the adapted flow is already maximum, i.e.\ the min cut does not change. 
			Using a binary search, we can find in $\mathcal{O}(\log(v_{\max}))$ time a point where the min cut changes.
			
			It is still possible that the left-most min cut changes, but the value of the min cut and thus of the flow remains the same. 
			Using \Cref{prop:shioura_monotonie}, we know that the left-most min-cut changes monotonically. 
			By the structure of the network, the left-most min cut is determined by the objects in the cut. 
			Thus, a left-most min cut with the same value can be seen at most $\abs{\Omega}$ times. 
			Thus, after $\mathcal{O}(\log(v_{\max}) \abs{\Omega})$ update steps, the flow needs to be augmentable.
		\end{subproof}
		
		This finishes the proof since the algorithm does at most $\mathcal{O}(\log(v_{\max}) \abs{\Omega} D_{\f{0}})$ update steps of cost $\mathcal{O}(\abs{N}\TO + \abs{N}\abs{\Omega})$ each.
	\end{proof}
 
	\subsection{VCG prices}\label{sec:VCG}
	When assuming that every object is owned by a seller and the revenue of the seller is the sum of prices he collects from buyers, then every Walrasian equilibrium is socially optimal (see e.g.~\cite{leme2017gross}).
	``Socially optimal" means that the utility of the buyers plus the utility of the seller is maximum among all allocation and price combinations  $(\f{p}, \f{x})$.
	This observation follows since in a Walrasian equilibrium $(\f{p}^*, \f{x}^*)$ every buyer maximizes her payoff given prices $\f{p}^*$. 
	The sum of the prices cancels out since the sellers get what the buyers pay. 
	Therefore, the Walrasian equilibrium is socially optimal.
	
	A famous mechanism in auction theory is the VCG mechanism (see for an introduction \cite{roughgarden2016cs269i} and \cite{vickrey1961counterspeculation,clarke1971multipart,groves1973incentives} for the original papers). 
	It takes bids or valuation functions of the buyers and computes an allocation and prices.
	The VCG mechanism is essential for auction theory, since it has two nice properties. 
	First, the computed allocation, is socially optimal and second, the mechanism is truthful. 
	This means, for a buyer it never pays off to lie, i.e., when reporting the true valuations a buyer $j$ is never worse off than with reporting a fake valuation. 
	Moreover, there exist some reported valuations of the other buyers where it is strictly better for buyer $j$ to report her true valuations than to lie.

	\paragraph{Single-unit matching markets.}
    As introduced in \Cref{sec:intro}, a famous special case of our model is the classical matching market (in the literature also commonly referred to as housing market) model, where buyers have unit demand.
    Interestingly, in this restricted model, the buyer-optimal Walrasian prices coincide with prices that are computed with the VCG mechanism (\cite{vickrey1961counterspeculation,clarke1971multipart,groves1973incentives}) using Clarke's pivot rule (see \cite{kern2016}). We call these prices \emph{VCG prices}.
    The VCG mechanism is known for the fact that it is a dominant strategy for all buyers to announce their valuations truthfully.
    However, note that this is only true for a sealed-bid auction and not necessarily an ascending auction, in which buyers might learn about other buyers' valuations as the auction progresses and are hence able to react and adjust their bidding strategy.
    Note that the ascending auction could be seen as a proxy auction to determine VCG prices when buyers submit their valuation for each item to a proxy bidder who answers all queries of the auctioneer truthfully.

    The property that minimal Walrasian prices coincide with VCG prices does not hold for markets where buyers have non-unit demand. 
    In \Cref{prp:noVCG}, we show that VCG prices in such markets might require that different buyers get charged different prices for the same object.
    Moreover, it has already been observed in \cite{gul2000english} that an auctioneer cannot get enough information in an ascending auction about buyers' valuation functions to determine VCG prices in those markets.
    
	We show that for truncated additive valuations, no prices that are based on a per object price can coincide with the VCG prices. 
	But before, we show that the very natural approach to copy buyers and give them unit demand does not lead to VCG prices. 
	
	Let us revisit \Cref{ex:1} from \Cref{sec:intro} and observe that the duplication method is not truthful, i.e., it pays off to lie, and thus does not lead to VCG prices. 
	For the buyer it would be strictly better to report the sets according to the valuation $v_{\bullet 1} = (1, 1)$. 
	
	\begin{proposition}\label{prp:noVCG}
        There is an example with buyers having truncated additive valuations where there exists no mechanism to determine prices per object that coincide with the VCG prices.
	\end{proposition}
	
	The following example shows that we need prices per bundle to obtain VCG prices, so the proposition holds.
	
	\begin{example}
		Let $N = \{1, 2, 3\}$ and $\Omega = \{\alpha, \beta\}$ with demands $d_1 = d_2 = 2, d_3 = 1$, supplies $b_\alpha = 3, b_\beta = 2$, and valuations $v_1 = (3, 1)$, $v_2 = (2, 0)$, and $v_3 = (0, 1)$. 
		It is easy to see that both feasible allocations (buyer~1 or~2 gets two items of object $\alpha$, the other one gets one item of each object, buyer~3 gets one item of object $\beta$ in both allocations) are socially optimal. 
		Say buyer~1 gets two items of $\alpha$, buyer 2 gets $\alpha$ only once. 
		Then the net effect of buyer~1 on buyer~2 (i.e., the cost that buyer~1 imposes on buyer~2 by her presence) is $2$, so the only possible price for object $\alpha$ is $1$. 
		Since buyer~2 also has to purchase an item of $\alpha$ once, object $\beta$ needs to be subsidized (i.e., assigned a price of $-1$) to give her a total cost of $0$. 
		However, this would give buyer 3 total cost of $-1$, which is not VCG.
	\end{example}

    As an immediate observation from the last example, we can also see that it would have been better for buyer~2 to submit her preferred bundles according to a false valuation function $\tilde{v}_2 = (0, 0)$, which would have resulted in prices $\f{p} = (0, 0)$ and the same allocation.
    Hence, she would have gotten a payoff of $2$.
    The prices we would have obtained by the auction with truthful reporting would have been $(2, 0)$ which yields a payoff of $0$ for buyer~2.
    Thus, the ascending auction does not incentivize truthful reporting of demand sets.
	
	\section{Conclusion and outlook}
	We present a network interpretation of an ascending auction for a multi-unit market where buyers have truncated additive valuations. 
	We show that by iteratively raising the prices on a left-most min cut we can compute buyer-optimal Walrasian prices via an ascending auction. 
	The new part here is the simple and efficient flow-based algorithm to determine the sets on which prices should be raised in the ascending auction, namely the minimal maximum overdemanded sets.
	For the special case of unit demand a nice matching based algorithm was known to compute these sets.
	For the general case of strong gross substitute valuations, in prior literature, this question was either not addressed or the computation was done by using tools from convex analysis.
	We currently work on a more natural and direct approach to compute the minimal maximum overdemanded sets also in the general setting.
	
	With our approach we are, moreover, able to reuse computations from previous iterations to speed up the computation. 
	Combining this with the algorithm using adapted step length allows for an improved running time analysis for the complete auction. 
	Still the resulting running time is not polynomial but only pseudo-polynomial. 
	While buyer-optimal Walrasian prices can be computed directly by using an LP, it remains open whether there is an iterative ascending auction that runs in polynomial time. Note that when using an LP approach, the buyers need to reveal their whole valuation function, while they only reveal their preferred bundles in the more natural iterative auction. 
	
	Our approach enabled us to show that the minimum Walrasian prices coincide with the minimum competitive prices for truncated additive valuations. 
	This connection seems very natural but was never discussed before. 
	It is the main part of our proof to show monotonicity properties, i.e., that buyer-optimal Walrasian prices can only increase when the supply decreases or the demand increases. 
	We are working on achieving a similar result for the general case of strong gross substitute valuations. 
	Also in this case the connection between minimum Walrasian prices and minimum competitive prices seems intuitive and would imply monotonicity results.
	
	Lastly, the number of iterations necessary to reach buyer-optimal Walrasian prices for a slightly perturbed instance is not polynomial bounded for the perturbations we considered. 
	It is an interesting question whether there are perturbations where the algorithm reaches buyer-optimal Walrasian prices again with only constantly many, or only polynomial many update steps.

\section*{Acknowledgments}
We would like to thank two anonymous reviewers who raised interesting questions and provided further insights, helping us to improve multiple sections of the article.

\bibliographystyle{splncs04}
\bibliography{literature}
	
\end{document}